\documentclass[acmsmall]{acmart}
\usepackage{float} %
\usepackage{booktabs} %
\usepackage{amsmath}
\usepackage{bm}
\usepackage{subfigure}
\usepackage{amsmath}
\usepackage[linesnumbered, boxed,ruled,lined]{algorithm2e}

\usepackage{amsmath}
\usepackage{multirow}
\usepackage{threeparttable}
\usepackage{graphicx}
\usepackage{verbatim}
\usepackage{color}
\usepackage{booktabs}
\usepackage{array}

\setcopyright{rightsretained}
\AtBeginDocument{%
  \providecommand\BibTeX{{%
    \normalfont B\kern-0.5em{\scshape i\kern-0.25em b}\kern-0.8em\TeX}}}

\setcopyright{acmcopyright}
\copyrightyear{2019}
\acmDOI{xx.xxxx/xxxxxxx.xxxxxxx}

\acmJournal{TOSN}
\acmVolume{X}
\acmNumber{XXX}
\acmArticle{1}
\acmMonth{9}
\acmYear{2019}

\begin{document}

\title{Computation Offloading with Multiple Agents in Edge Computing-supported IoT}

\author{Shihao Shen}
\authornotemark[2]
\affiliation{%
	\institution{Tianjin University}
	\city{Tianjin}
	\country{China}}
\email{shenshihao@tju.edu.cn}

\author{Yiwen Han}
\authornotemark[2]
\affiliation{%
	\institution{Tianjin University}
	\city{Tianjin}
	\country{China}}
\email{hanyiwen@tju.edu.cn}

\author{Xiaofei Wang}
\authornote{The corresponding author.}

\authornote{Tianjin Key Laboratory of Advanced Networking (TANK), College of Intelligence and Computing, Tianjin University.}
\affiliation{%
  \institution{Tianjin University}
  \city{Tianjin}
  \country{China}
}
\email{xiaofeiwang@tju.edu.cn}

\author{Yan Wang}
\authornote{The National Key Laboratory of Underwater Acoustic Science and Technology, College of Underwater Acoustic Engineering, Harbin Engineering University.}
\affiliation{%
	\institution{Harbin Engineering University}
	\city{Harbin}
	\country{China}
}
\email{wangyan@hrbeu.edu.cn}

\renewcommand{\shortauthors}{Shihao Shen, et al.}

\thanks{This work is supported by the National Key R$\& $D Program (2018YFC0809803, 2019YFB2101901) of China, China NSFC (Youth) through grant 61702364, Huawei Innovation Research Program (HO2018095224) and the Opening Fund of Acoustics Science and Technology Laboratory(Grant No.SSKF2018002). This paper is an extension of our conference paper published in the ACM TURC (DOI: 10.1145/3321408.3321586). The authors extend the previous work by improving the expression of parameters, adding complexity analysis to the problem, discussing the theoretical analysis of the proposed algorithm and adding multi-angle fine-grained related experiments.
	The authors would like to thank the anonymous referees for their valuable comments and helpful suggestions.
	\noindent\rule[0.25\baselineskip]{13.9cm}{0.5pt}}

\begin{abstract}
	With the development of the Internet of Things (IoT) and the birth of various new IoT devices, the capacity of massive IoT devices is facing challenges. Fortunately, edge computing can optimize problems such as delay and connectivity by offloading part of the computational tasks to edge nodes close to the data source. Using this feature, IoT devices can save more resources while still maintaining the quality of service. However, since computation offloading decisions concern joint and complex resource management, we use multiple Deep Reinforcement Learning (DRL) agents deployed on IoT devices to guide their own decisions. Besides, Federated Learning (FL) is utilized to train DRL agents in a distributed fashion, aiming to make the DRL-based decision making practical and further decrease the transmission cost between IoT devices and Edge Nodes. In this paper, we first study the problem of computation offloading optimization and prove the problem is an NP-hard problem. Then, based on DRL and FL, we propose an offloading algorithm that is different from the traditional method. Finally, we studied the effects of various parameters on the performance of the algorithm and verified the effectiveness of both the DRL and FL in the IoT system.
\end{abstract}



\maketitle
\section{Introduction}

The Internet of Things(IoT) is an extension of the Internet. It can connect things to the network for communication by using sensors to realize intelligent management, positioning, identification, monitoring, and other functions. At present, IoT is rapidly emerging and has driven the vigorous development of various related application services. It has been applied in many fields such as smart home~\cite{jie2013smart}, human health detection~\cite{islam2015internet}, disaster management~\cite{kamruzzaman2017study}, building structure safety~\cite{hu2013wireless}, person identificationand~\cite{khalil2018sonicdoor} and so on. 

As shown in Fig. \ref{fig:app_iot}, IoT has been closely integrated with various fields. However, various application services need to deploy a large number of IoT devices while providing rich functions, which puts heavy pressure on communication. Besides, various new IoT devices such as smart home devices and wearable devices are also emerging and they have extremely stringent requirements on the bandwidth, delay, and privacy of the network~\cite{samuel2016review,pyattaev2015communication}, which pose challenges to the quality of the communication.
\begin{figure}[htbp]%
	\centering
	\includegraphics[width=0.7\textwidth]{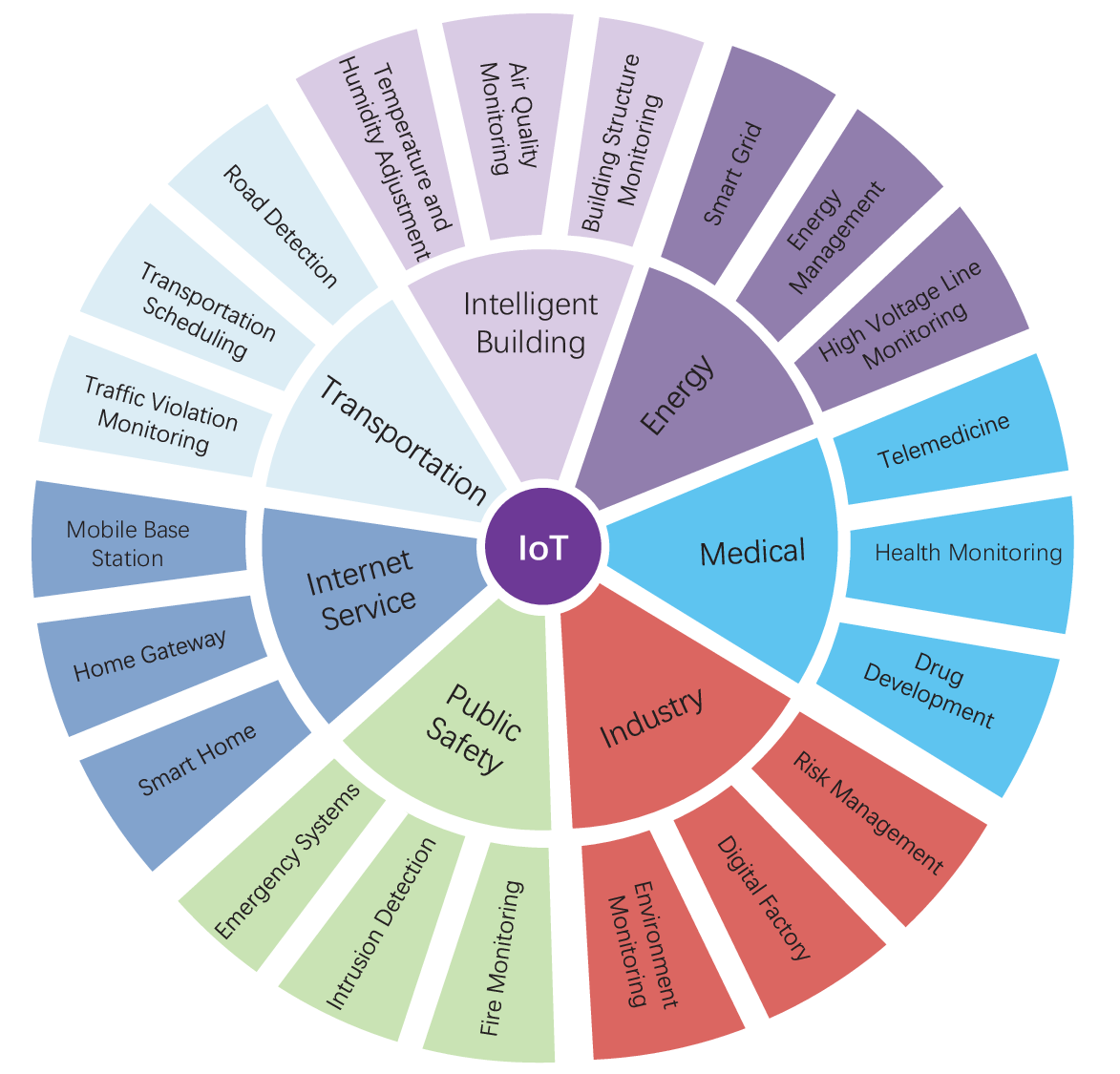}
	\caption{The applications of IoT.} 
	\label{fig:app_iot}
\end{figure}

Catering for the increasing service requirement, massive IoT devices will be deployed in a pervasive fashion to carry out tasks like monitoring, sensing data collection and preprocessing and immediate decision-making. The above tasks usually require a large number of computing resources, while the ability of IoT devices is relatively weak to support. However, edge computing can offload tasks and is expected to solve this problem~\cite{Shi2016, Mao2017A}. In detail, edge computing can offload computing tasks submitted by IoT devices to similar edge nodes to provide rich computing resources. Besides, the edge node in edge computing systems is adopted as a coordinator among them and responsible for their communications and even load-balance~\cite{wang2018secure}.

In solving the resource allocation problem in task offloading, in addition to using convex optimization~\cite{Chen2018d} and game theory~\cite{Chen2016f}, Deep Reinforcement Learning (DRL) is used in \cite{Chen2018k} to handle the comprehensive resource allocation in computation offloading. This approach can maximize the long-term benefits of energy consumption and execution latency and requires no prior knowledge of network statics and partial information. In detail, this kind of optimization has many advantages. First, the IoT device does not need to obtain global information, which is conducive to communication transmission and privacy protection. Secondly, it has the adaptability to the dynamic environment. Finally, this kind of optimization not only optimizes the system within a time segment but will consider long-term benefits. However, an assumption was made in \cite {Chen2018k}. They assume that IoT devices are computationally powerful enough to train their own DRL agents independently. However, IoT devices will not grow so powerful in the near future, and their computing resources can support lightweight neural networks at most.

At present, due to people pay more attention to data security and privacy, the protection of data and privacy has become an important issue that must be considered~\cite{yang2017survey}. For example, the General Data Protection Regulation~\cite{GDPR} implemented in the European Union is aimed at data security and privacy, and it gives users the right to delete or withdraw personal data. Therefore, the traditional way of transferring data to a data center for centralized analysis will face a privacy barrier in the future. In summary, how to protect data security and privacy while using large amounts of data will become an important challenge in the future.

Thus, we propose a distributed training scheme based on Federated Learning (FL)~\cite{mcmahan2016communication,bonawitz2019towards} to alleviate the training burden on each device. Unlike the traditional distributed training in the data center with an excellent networking environment, this training is restrained by wireless communication and networking and shall be performed in an efficient manner of communication. In this vein, observation data sensed by each IoT device are not required to be transmitted frequently between it and the edge node. Observation data on a specific IoT device are used for local training, and only updated parameters of the DRL agent are uploaded to the edge node for further model aggregation.

Therefore, in this work, we use FL to conduct the training process of DRL agents for jointly allocating communication and computation resources. Specifically, our main contributions lie in three aspects.

\begin{itemize}
	\item First, we studied a problem of computation offloading optimization in edge computing-supported IoT. On one hand, the computational tasks can be performed locally, while some energy needs to be allocated to the task processing components in the IoT device. On the other hand, it can also be performed by transmitting tasks to the edge nodes, while requiring some energy to be allocated to the data transmission components in the IoT device. Compared to local execution, this can be done with the richer computing resources of the edge nodes, but it also causes additional transmission delays due to data transmission. Moreover, we further analyzed the complexity of the problem and proved that the problem is NP-hard.
	\item Second, we designed an algorithm to make decisions about computation offloading and energy allocation, seeking to maximize the expected long-term utility. This algorithm can be trained based on federated learning, so data collected by each IoT device only needs to be stored locally for analysis. This approach avoids a lot of data transmission and achieves good data privacy protection.
	\item Finally, we simulate to evaluate the proposed algorithm and study the influence of various system parameters. Experimental results corroborate its effectiveness comparing to the centralized training method.
\end{itemize}

\section{Background}
\label{sec:technical foundation}
In this section, we introduce the background of this study. Since the technology that this research relies on mainly involves deep reinforcement learning, federated learning and edge computing, the principles and related research of these technologies are briefly introduced.

\subsection{Deep Reinforcement Learning}
\label{sec:deep reinforcement learning}
Since reinforcement learning (RL) techniques are usually applied to small data spaces, it is difficult to perform data processing through RL when data with high dimensions. However, Deep reinforcement learning (DRL) solves this problem by combining the high-dimensional input of deep learning with RL.

RL usually attempts to make action decisions $a$ according to the environment state $s$, and the action proceeds to the environment to obtain the action reward $r$, and continuously adjusts and improves according to $r$~\cite{sutton2018reinforcement}. While deep learning is a method of characterizing and learning data through a multi-layer neural network and learning the characteristic information of the data through the neural network. DRL combines deep learning and RL which not only retains the perception of deep learning but also can make decisions for RL, so it has better performance. 

DRL has been well applied in many fields such as natural language processing~\cite{shah2018follownet}, image recognition~\cite{rao2017attention}, and so on. Among them, AlphaGo~\cite{silver2016mastering}, which defeated the human professional chess player in go and showed excellent decision-making ability. Besides, some researchers use DRL to play Atari. They train the model by using the joystick moving direction as the action space of DRL and using the score in the game as a reward. In the game, DRL surpassed the traditional method in six of the games, even in the performance of three games beyond the human level. Besides, various DRL libraries such as TensorFlow~\cite{tflite}, Caffe~\cite{jia2014caffe} and Keras~\cite{Keras} are also emerging, which facilitates the application of DRL.

\begin{figure}[htbp]
	\centering
	\includegraphics[width=0.8\textwidth]{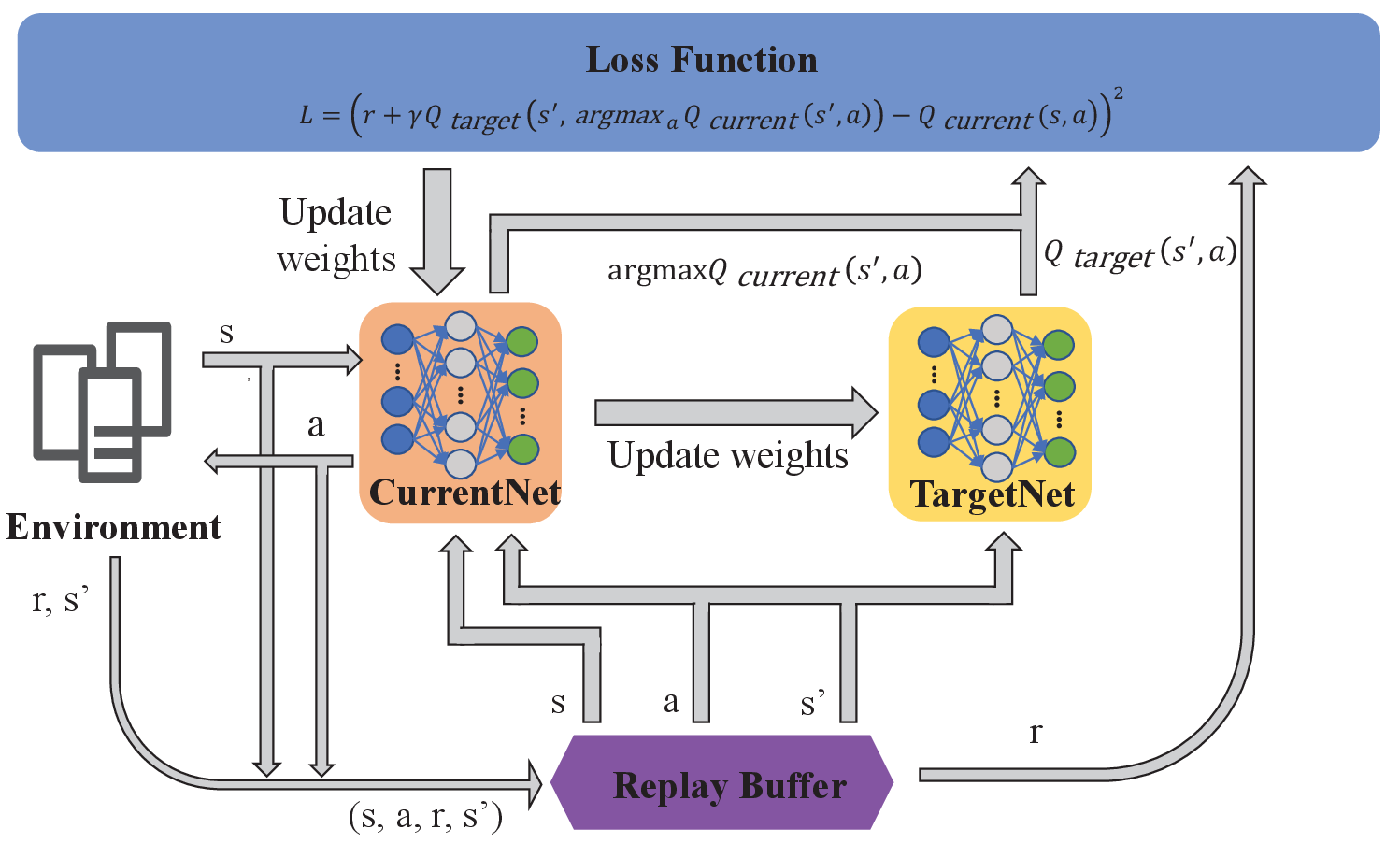}
	\caption{ Diagram of DDQN.}
	\label{fig:ddqn}
\end{figure}

As shown in Fig. \ref{fig:ddqn}, Double Deep Q-Learning (DDQN)~\cite{Hasselt20016} is an excellent DRL algorithm. To solve the problem of the curse of dimensionality in reinforcement learning, DDQN uses a neural network to approximate some states that have not appeared before:
\begin{equation}
	\label{Q approximate}
	Q(s, a) \approx f(s, a, w)~.
\end{equation}

The training process of the neural network is an optimization problem, so the loss function is defined as:
\begin{equation}
	L=(r+\gamma Q_{\text { target }}\left(s^{\prime}, \mathop{\arg\max}_{a} Q_{\text { current }}\left(s^{\prime}, a\right)\right)-Q_{\text { current }}(s, a))^{2}~.
	\label{loss_function}
\end{equation}

In addition, since the training of the neural network is supervised, the training data must satisfy the independent and identical distribution, otherwise, the network will be trapped in the local minimum. Therefore, a replay buffer $\mathcal{B}$ is constructed to store the data sample $D_t = (s_t, a_{t}, r_{t}, s_{t+1})$ at each time step $t$, and randomly extract a mini-batch of samples during training.

\subsection{Federated Learning}
\label{sec:federated learning}
The traditional large-scale neural network training needs to concentrate data in one device, which puts great challenges on traffic load and data privacy. In this regard, Google has proposed Federated learning~\cite{federatedlearning}. It allows multiple end devices to train on local data, and then only needs to upload updates to the cloud.
\begin{figure}[htbp]
	\centering
	\includegraphics[width=0.8\textwidth]{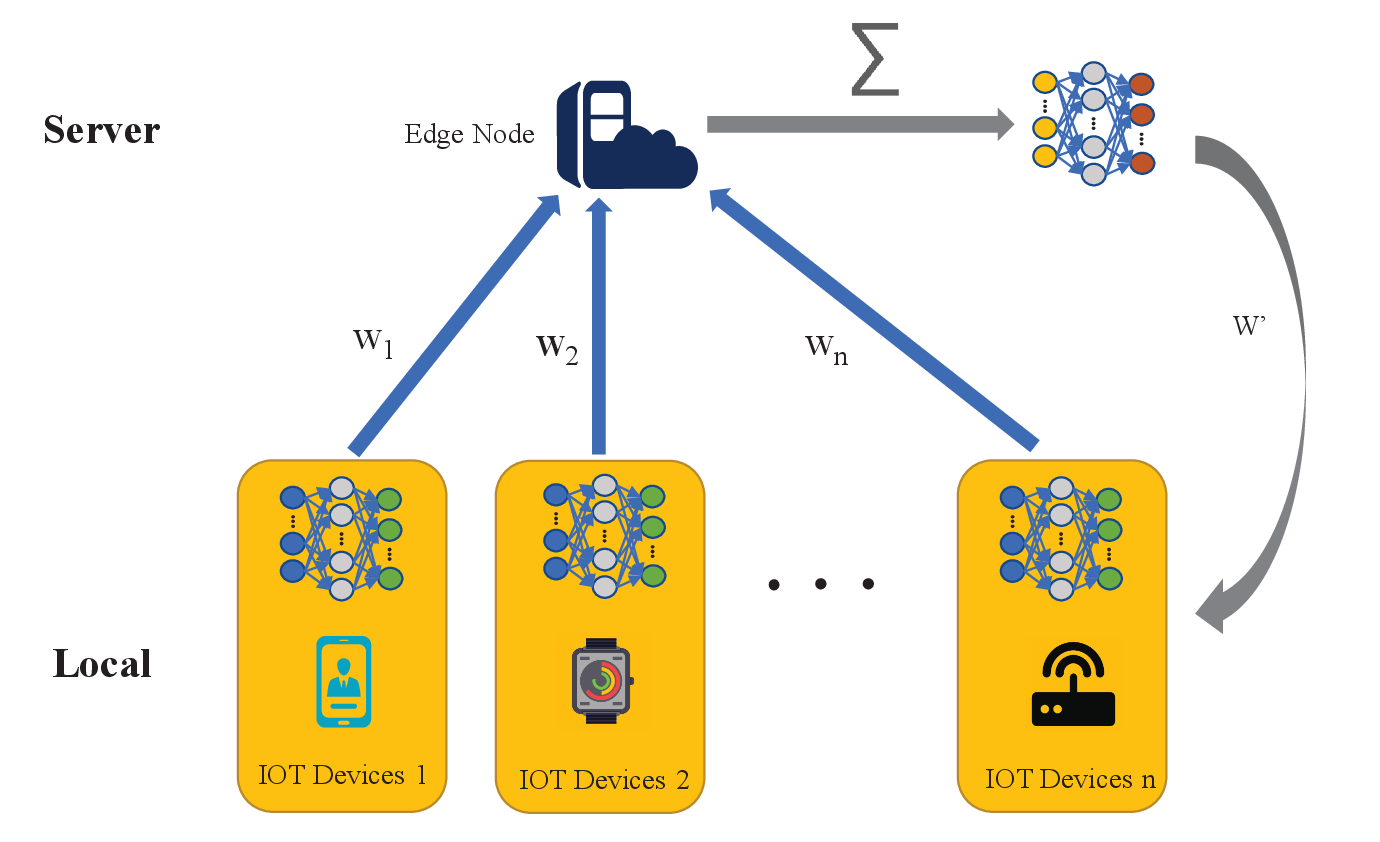}
	\caption{The training process of federated learning.}
	\label{fig:fl}
\end{figure}

The way of federated learning works is shown in Fig. \ref{fig:fl}. First, the end devices download the sharing model from the cloud, and then train the model according to the local data and transfer the update to the cloud with encrypted transmissions. Finally, the cloud integrates the sharing model according to the update from multiple end devices. Since user data is always stored locally on the end device throughout the process, large amounts of data can be avoided from being transmitted to the cloud, thereby reducing the pressure on data transmission and protecting data privacy.

In the practical application of federated learning, there are still some problems. On the one hand, sample data will be distributed in a large number of end devices in an extremely uneven manner. On the other hand, the transmission speed of end devices is slow, especially data uploading speed can limit overall performance. To solve these problems, Google developed an algorithm federated averaging to reduce the network requirements when training deep neural networks~\cite{mcmahan2016communication}, and also compress updates by using random rotation and quantization to reduce the amount of data transferred~\cite{konevcny2016federated}. In addition, a federate optimization algorithm was designed to optimize the high-dimensional sparse convex model~\cite{konevcny2016federated2}.

\subsection{Edge Computing}
\label{sec:edge computing}
Edge computing provides network, computing, application, and storage services to close users through distributed edge nodes. Therefore, tasks can be performed on EN to avoid sending data to the cloud. In 2014, the European Telecommunications Standards Institute standardized the concept of edge computing~\cite{ETSI2}, which also marked the standardization of edge computing technology. 

The end devices in edge computing are diverse, such as connected vehicle~\cite{grewe2017information}, smart cameras~\cite{ananthanarayanan2017real}, etc. and they are producers of data and tasks with data pre-processing and data transmission functions~\cite{abbas2017mobile}. However, when an end device needs to handle a task with a very large computational resource, it is often difficult to rely on the computing capability of the device itself to meet the demand. Therefore, it can be solved by edge computing using the computing resources of the edge node. The edge node is geographically close to the end device and can provide high-quality network connection and computing services. Compared with the end device, the edge node has a more powerful computing capacity to process the task, and the edge node responds faster to the end device than the cloud. Therefore, by using edge nodes to perform some computational tasks, the response speed of the task can be improved while ensuring accuracy. In addition, the edge node also has a caching function~\cite{Zeydan2016}, which can shorten the response time of re-access by caching objects with high access heat.

\section{System Model}
\label{sec:system_model}
Next, we will introduce the system model used in this study. First, the overall architecture of the edge computing-supported IoT system is presented and the relevant parameters involved in the system are introduced. After that, we introduce the changing ways of the related parameters and show the derivation process of the parameters.

\subsection{Overview}
\label{sec:system_architecture}

\begin{figure}[htbp]
	\centering
	\includegraphics[width=0.8\textwidth]{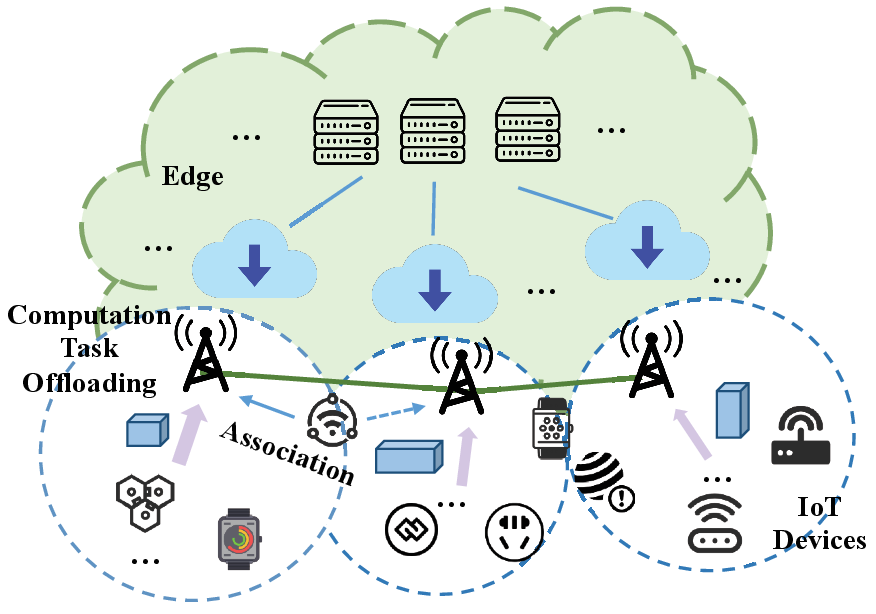}
	\caption{Edge computing-supported IoT system.}
	\label{fig:iot_edge_system}
\end{figure}

As shown in Fig. \ref{fig:iot_edge_system}, IoT devices, denoted as $\mathcal { D } = \{ 1 , \cdots , D \}$, are located in the service area of a set $\mathcal { N } = \{ 1 , \cdots , N \}$ of Edge Nodes (ENs). It is worth noting that the concept of time slices is used to divide the time into several epochs of $\delta$ (in seconds) indexed by $i$. The functions of transmission and calculation are supported by EN and the geographical location, task processing capability and data transmission capability of each EN are different. For IoT devices, there will be a task queue with a maximum length of $q ^ { \mathrm { t } } _ { \mathrm { max } }$ to temporarily store tasks and these tasks will be executed in the order of first-in-first-out (FIFO). The IoT device will generate tasks according to the Bernoulli distribution and define $a_{i}^{\mathrm { t }}$ as the task arrival indicator. $a_{i}^{\mathrm { t }}=1$ indicates that the epoch $i$ has a task generation, otherwise it indicates that no task is generated. When the task queue has reached the upper limit, the task queue will not store the newly generated task, which will cause the task to directly fail. In addition, the collection method of energy units is similar to that discussed in ~\cite{adu2018energy} and the IoT devices can collect energy units from the outside. In modeling, the IoT device has an energy queue with a maximum length of $q ^ { \mathrm { e } } _ { \mathrm { max } }$ to store energy units, which will acquire the energy unit according to the Poisson distribution.

The calculation task generated by IoT devices is modeled as $(\mu, \nu)$, where $\mu$(in bits) and $\nu$ respectively represent  the transmission data size required for offloading a task and the needed number of CPU cycles for processing the task. Moreover, there are two ways to accomplish these calculation tasks, one is performed locally on the IoT device, and the other is offloading to an EN with channel bandwidth $W$~Hz to perform. However, as shown in Fig. \ref{fig:edge_local}, the delay caused by data transmission can be avoided if the task is executed locally. For another, if the task is executed at EN, the task can be performed with the richer computing resources of EN. Both methods have their advantages and disadvantages and they need to be weighed according to the specific state. 

\begin{figure}[htbp]
	\centering
	\includegraphics[width=0.8\textwidth]{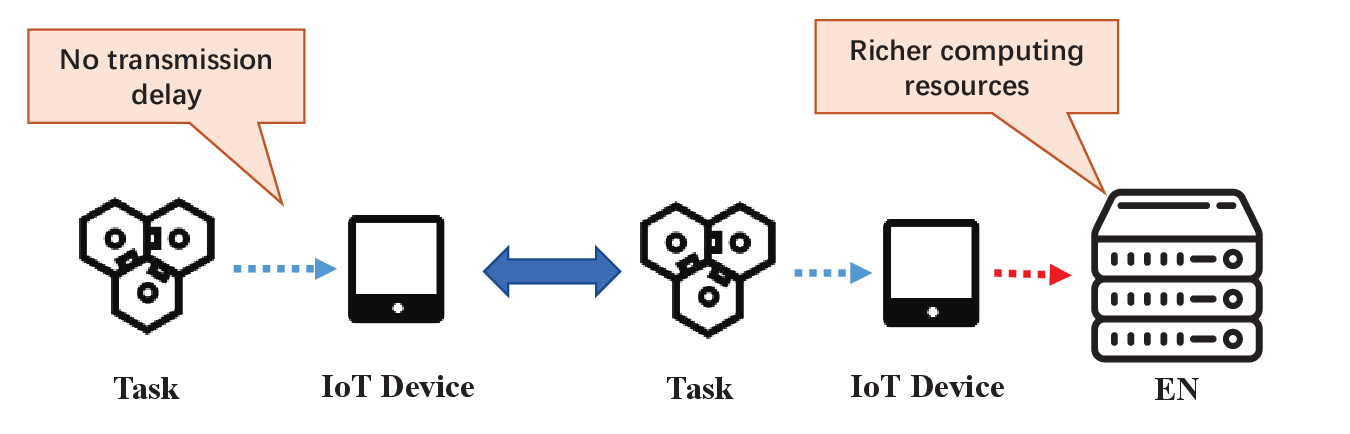}
	\caption{Task execution method}
	\label{fig:edge_local}
\end{figure}
\subsection{Description of System Model}
\label{sec:Description}
\subsubsection{System model architecture}
Since the execution mode of the task in the system model includes local execution and offloading to EN, the dynamic running process of the system is as shown in Fig. \ref{fig:dynamic} .
\begin{figure}[htbp]
	\centering
	\includegraphics[width=1.0\textwidth]{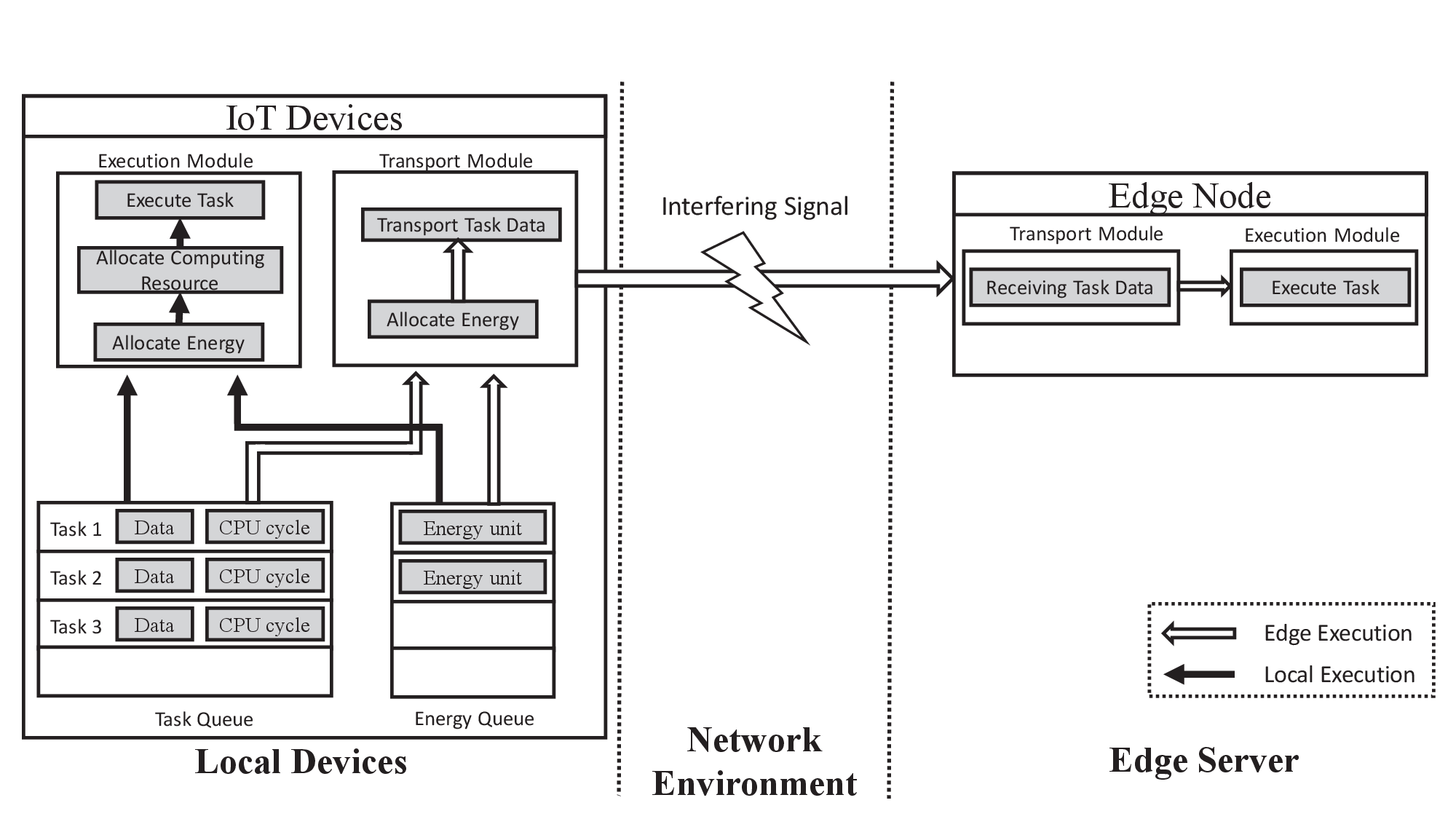}
	\caption{Schematic diagram of the dynamic system model.}
	\label{fig:dynamic}
\end{figure}

The IoT device needs to make a joint action $\left( c _ { i } , e _ { i } \right)$ during the epoch $i$, where $ c _ { i } $ represents the task offloading decision, and its specific definition is
\begin{equation}
	\label{offload_decision_description}
	c_{i} = \left\{ \begin{array} { l l } 
		{ 0 , } & { \text { if }  \text{local execution,} } \\ 
		{ n \in \mathcal{N} , } & { \text { if } \text{offload to EN n.}}  \end{array} \right.
\end{equation}

Besides, $ e _ { i } $ represents the number of energy units allocated and affects the CPU frequency and data transmission rate of IoT devices. Moreover, $e _ { i }$ cannot exceed the number of energy units in the energy queue, and if $e _ { i } = 0$, the task will not be executed and will still be saved in the task queue, and only $e _ { i } > 0$, the task will be executed. In addition, IoT devices also have a task queue that will not be able to save newly created tasks when the task queue is full. 

\subsubsection{Local execution}

When the task is executed locally, the time consumption satisfies the following constraints
\begin{equation}
	\label{time_consumption}
	d ^ { \mathrm { m } } _ { i } =  { e _ { i } } / {p^{exe}_{i} } ~,
\end{equation}
where $p^{exe}_{i}$ is the power of the IoT device to perform a task. Meanwhile, as in \cite{gerards2014interplay}, $p^{exe}_{i}$ can be written as
\begin{equation}
	\label{p_exe}
	p^{exe}_{i} =  \tau \cdot f _ { i }^{\zeta} ~,
\end{equation}
where $\tau$ is a constant that depends on the average switched capacitance and $\zeta$ is the average activity factor which usually close to 3. Besides, the time consumption also satisfies the following constraints
\begin{equation}
	\label{constraints_f}
	d ^ { \mathrm { m } } _ { i }  =  \nu / f _ { i } ~.
\end{equation}

From the above, we can get the time consumption $d ^ { \mathrm { m } }_ { i }$  for local execution by solving equation (\ref{time_consumption}), (\ref{p_exe}) and (\ref{constraints_f}), which can be written as
\begin{equation}
	\label{result_dm}
	d ^ { \mathrm { m } } _ { i }  =  (\dfrac{\nu^{\zeta} \cdot \tau}{e_i})^{\frac{1}{\zeta-1}} ~.
\end{equation}

\subsubsection{Executed at EN}
If the IoT device decides to offload the task to EN, then the relevant data of the task needs to be transferred to EN. Therefore, an association between the IoT device and the EN is required. If the required association is different from the previous, a handover will occur and cause additional delays.  Let $s_i$ denote the association between IoT device and EN at the beginning of epoch $i$, which can be written as
\begin{equation}
	\label{association}
	s _ { i } = \left\{ \begin{array} { l l } 
		{ s_{i-1} , } & { \text { if }  \text{executed locally in the previous epoch ($c _ { i-1 }$ = 0),}} \\ 
		{ c_{i-1} , } & { \text { if }  \text{offload to EN in the previous epoch ($c _ { i-1 } \ne $0).}}  \end{array} \right.~
\end{equation}

Moreover, the handover delay resulting from altering EN association should be considered. We assume that there will be $\sigma$ seconds delay when handoff occurs, so the handover delay $h _ {i}$ can be represented as
\begin{equation}
	\label{hi}
	h _ { i } = \left\{ \begin{array} { l l } 
		{ 0, } & { \text { if }  \text {no altering EN association ($c _ { i } = s _ { i }$),} } \\ 
		{ \sigma , } & { \text { if }  \text {altering EN association ($c _ { i } \ne s _ { i }$).} }  \end{array} \right.~
\end{equation}

In addition, it is necessary to model the rate of transfer between the IoT device and EN when data transfer occurs. We denote $g ^ { n } _ { u }$ as the channel gain of the IoT device $u$ between the IoT device and an EN $n \in \mathcal { N }$, which is assumed static and independently taken from a finite state space $\mathcal{G} _ n$. Let $L_{c}$ denote all sets of IoT devices that use the same channel as the target IoT device $u_{t}$. When a radio link is established for them, the achievable data rate can be calculated as
\begin{equation}
	\label{data_rate}
	r _ { i } = W \cdot \log _ { 2 } \left( 1 +  \dfrac{g ^ { n } _{u_{t}} \cdot p ^ { \mathrm { tr } } _{u_{t}}}{\sum_{u \in L_{c}} g^{n}_{u} \cdot p_{u}^{\mathrm{tr}}-g^{n}_{u_{t}} \cdot p_{u_{t}}^{\mathrm{tr}}} \right)~,
\end{equation}
where $p_{u}^{\mathrm{tr}}$ represent the transmit power of the IoT device $u$ and $p_{u}^{\mathrm{tr}}$ satisfies
\begin{equation}
	\label{trans_power}
	p_{u} ^ { \mathrm { tr } }  = {e _ { i }} / { d ^ { \mathrm { tr } } _ { i } } \leq p ^ { \mathrm { tr } } _ { \mathrm { max } }.
\end{equation}

Thus, the time consumption on data transmission can be expressed as
\begin{equation}
	\label{trans_time}
	d ^ { \mathrm { tr } } _ { i } = \mu /  { r _ { i } }~.
\end{equation}

According to the proof in \cite{Chen2018k}, given the association $s _ { i } \in \mathcal { N }$ and the allocated energy units $e _ {i} > 0$ at a epoch $i$, the transmitting rate should remain a constant for achieving the minimum transmission time, which is preferred in practical. Therefore, the minimum transmission time can be solved by equations of (\ref{data_rate}), (\ref{trans_power}) and (\ref{trans_time}) as
\begin{equation}
	\label{solve_trans_time}
	\log _ {2} ( 1 + \frac{ g_{i}^{c _ i} \cdot e_{i} }{  d^{\text{tr}}_i \cdot \sum_{u \in L_{c}} g^{n}_{u} \cdot p_{u}^{\mathrm{tr}}-g^{n}_{u_{t}} \cdot p_{u_{t}}^{\mathrm{tr}} } ) = \frac{\mu}{W \cdot d_{i}^{\text{tr}} }~.
\end{equation}

After the task is offloaded to an EN, the task will be completed by this edge node. While the execution delay $d ^ {s}$ of a task in EN is much less than the transmission delay $d ^ { \mathrm { tr } } _ { i }$, so $d ^ {s}$ is set to a small constant. In addition, the payment $\phi _ { i }$ of occupying the EN is set to avoid excessive use of EN resources in actual operation. With defining $\pi \in \mathbb { R } _ { + }$ as the price per unit of time, the payment expression can be written as
\begin{equation}
	\label{offload_payment}
	\phi _ { i } = \pi \cdot \left( \min \left\{ h _ { i } + d ^ { \mathrm { tr } } _ { i } + d ^ {s}, \delta \right\} - h _ { i } \right)~.
\end{equation}

\subsubsection{Update system model parameters}
Through the above modeling of different treatments of tasks, the task execution delay can be summarized as follows:
\begin{equation}
	\label{execution_delay}
	d _ { i } = \left\{ \begin{array} { l l } { d ^ { \text { m } } _ { i } , } & { \text { if } \text {local execution ($e _ { i } > 0 \text { and } c _ { i } = 0$),} } \\
		{ h _ { i } + d ^ { \text { tr } } _ { i } + d ^ { \text { s } }, } & { \text { if }  \text {offload to EN to execute ($e _ { i } > 0 \text { and } c _ { i } \in \mathcal { N }$), }} \\
		{ 0 , } & { \text { if } \text {not executed ($e _ { i }$ = 0).} } \end{array} \right.~
\end{equation}

In addition, it is also considered that not all tasks can be executed immediately after they are generated, so let $\rho _ { i }$ denote the queuing delay in the task queue at epoch $i$, which can be described as
\begin{equation}
	\label{queuing_delay}
	\rho _ { i } = q ^ { \mathrm { t } } _ { i } - 1 _ { \left\{ d _ { i } > 0 \right\} }~.
\end{equation}

In each epoch, we need to pay attention to the dynamic changes of the task queue and energy queue of the IoT device. In the details, the change in the length $q _ {i} ^ {\text{e}}$ of the energy queue can be described as
\begin{equation}
	\label{energy_queue_dynamic}
	q ^ { \mathrm { e } } _ { i + 1 } = \min \left\{ q ^ { \mathrm { e } } _ { i } - e _ { i } + a ^ { \mathrm { e } } _ { i } ,~ q ^ { \mathrm { e } } _ { \max } \right\}~,
\end{equation}
where $ a _ {i} ^ {\mathrm{e}} \in \mathbb { N } _ { + } $ is the number of energy units acquired by the IoT device in epoch $i$. For task queues, we need to consider the generation and completion of tasks for each epoch. Similarly, let $a ^ { \mathrm { t } } _ { i }$ denote the number of tasks generated in epoch $i$. Then dynamics of the task queue length can be calculated as
\begin{equation}
	\label{task_queue_dynamic}
	q ^ { \mathrm { t } } _ { i + 1 } = \min \left\{ q ^ { \mathrm { t } } _ { i } - \mathbf { 1 } _ { \left\{ 0 < d _ { i } \leq \delta \right\} } + a ^ { \mathrm { t } } _ { i }~~,~~q ^ { \mathrm { t } } _ { \mathrm { max } } \right\}~.
\end{equation}

However, newly generated tasks cannot be stored and fail directly when the task queue is full. So let $\eta _ { i }$ denote the number of computation task drop in an epoch $i$, which can be described as
\begin{equation}
	\label{task_drop_num}
	\eta _ { i } = \max \left\{ q ^ { \mathrm { t } } _ { i } - \mathbf { 1 } _ { \left\{ 0 < d _ { i } \leq \delta \right\} } + a ^ { \mathrm { t } } _ { i } - q ^ { \mathrm { t } } _ { \max } , ~0 \right\}~.
\end{equation}

In order to express this model more clearly, the values and definitions of the parameters are shown in table \ref{definitions}.

\begin{table*}[htbp]  
	\centering  
	\fontsize{6.5}{8}\selectfont  
	\begin{threeparttable}  
		\caption{The values and definitions of main parameters}
		\label{definitions}  
		\begin{tabular}{m{1cm}<{\centering}m{6cm}<{\centering}m{2.5cm}<{\centering}}  
			\toprule          
			\multicolumn{1}{c}{\bf Parameters }&\multicolumn{1}{c}{\bf Definitions}&\multicolumn{1}{c}{\bf Values}\cr
			\midrule 
			$\mathcal { D }$&The set of IoT devices.&/\cr
			$\mathcal { N }$&The set of Edge Nodes.&/\cr
			$W$&Channel bandwidth.&$6.0 \times 10^{5} \mathrm{Hz}$ \cr
			$\delta$&Duration of each epoch.&$5.0 \times 10^{-3} \mathrm{s}$ \cr
			$a_{i}^{\mathrm { t }}$&Task arrival indicator.&/\cr         
			$\mu$&The transmission data size required for offloading a task.&$3.0 \times 10^{4} \mathrm{bit}$\cr 
			$\nu$&The needed number of CPU cycles for processing a task.&$8.375\times 10^{6} \mathrm{cycle}$\cr 
			$c_i$&Task offload decision.&/\cr
			$e_i$&The number of energy units allocated.&/\cr
			$\tau$&A constant about the average switched capacitance. &$1.0 \times 10^{-28} $\cr
			$f_i$&Allocated CPU frequency of the IoT device.&/\cr
			$f_{\mathrm{max}}^{c}$&The maximum CPU frequency of the IoT device.&$2.0 \times 10^{9} \mathrm{Hz}$\cr
			$d_i^m$&Time consumption for the local task execution.&/\cr
			$r_i$&The achievable data rate of IoT device.&/\cr
			$p_{\mathrm{max}}^{\mathrm{tr}}$&The maximum transmit power.&$2 \mathrm{W}$\cr
			$p_{i}^{\mathrm{tr}}$&The maximum transmit power at epoch $i$.&/\cr
			$d_i^{\mathrm{tr}}$&The time for transmitting task data.&/\cr
			$q^e_{i+1}$&The energy queue length insider the IoT device at epoch $i$.&/\cr
			$d^s$&The delay of server-side execution.&$1.0 \times 10^{-6} \mathrm{s}$\cr
			$q_{\mathrm{max}}^{\mathrm{t}}$&The maximum length of the local task queue.&4\cr
			$q_{\mathrm{max}}^{\mathrm{e}}$&The maximum length of the local energy queue.&4\cr
			$d_i^{\mathrm{tr}}$&The time for transmitting task data.&/\cr
			$\sigma$&Delay of one handover.&$2.0 \times 10^{-3} \mathrm{s}$\cr
			$h_i$&The handover delay resulting from altering EN association.&/\cr
			$q^t_{i}$&Task queue length at epoch $i$.&/\cr
			$\eta_{i}$&The number of computation task drop at epoch $i$.&/\cr
			$\rho_{i}$&Queuing delay during the epoch $i$.&/\cr
			$\varphi_{i}$&The penalty of a computation task fails.&/\cr
			$\phi_{i}$&The payment of occupying the EN.&/\cr	
			\bottomrule  
		\end{tabular}  
	\end{threeparttable}  
\end{table*}

\section{Policy Training coordinated by Federated Learning}
In this section, we first explain the optimization problem to be solved and formulate the problem. In addition, we further analyze the complexity of the problem and prove that the optimization problem is NP-hard. After that, we analyze the merits of federated learning in edge computing and propose an algorithm of federated learning-based policy training. Finally, we carry out a theoretical analysis of the proposed algorithm.

\subsection{Problem Formulation}
Based on the system model described in Sec. \ref{sec:system_architecture} and Sec. \ref{sec:Description}, we will discuss the optimization problem next. First we define $\mathcal{X} _ { i }$ to represent the network environment of the IoT device during the epoch $i$.

\begin{equation}
	\label{equ:network_states}
	\begin{aligned}
		\mathcal{X} _ { i } &= ( q ^ { \mathrm { t } } _ { i } , q ^ { \mathrm { e } } _ { i } , s _ { i } , \bm{ g } _ { i } ) \in \bm{\mathcal{X}} \\
		& \stackrel { \mathrm { def } } { = } \left\{ 0,1 , \cdots , q ^ { \mathrm{t} } _ { \max } \right\} \times \left\{ 0,1 , \cdots , q ^ { \mathrm { e } } _ { \max } \right\} \times \mathcal { N } \times \left\{ \times _ { n \in \mathcal { N } } \mathcal { G } _ { n } \right\},
	\end{aligned}~
\end{equation}
where $\bm{ g } _ { i } = \left( g ^ { n } _ { i } : n \in \mathcal { N } \right)$. The IoT device will make the offloading decision and determine the number of energy units allocated during the initial period of epoch $i$, i.e.,

\begin{equation}
	\label{equ:joint_action}
	\left( c _ { i } , e _ { i } \right) \in \mathcal { J } \stackrel { \mathrm { def } } { = } \{ \{ 0 \} \cup \mathcal { N } \} \times \left\{ 0,1 , \cdots , q ^ { \mathrm { e } } _ { \max } \right\}~.
\end{equation}

The policy for making the above actions is defined as $\bm { \Phi }$ and the expected long-term utility is defined as
\begin{equation}
	\label{long_term_utility}
	U ( \mathcal{X} , \boldsymbol{\Phi} ) = \mathbb{E} _ { \boldsymbol{\Phi} } \left[ \lim _ { I \rightarrow \infty } \frac { 1 } { I } \cdot \sum _ { i = 1 } ^ { I } u \left( \mathcal{X} _ { i } , \Phi \left( \mathcal{X} _ { i } \right) \right) | \mathcal{X} _ { 1 } = \mathcal{X} \right]~,
\end{equation}
where $\mathcal{X} _ { 1 }$ represents the initial network environment and  $u(\cdot)$ represents the short-term utility in the epoch $i$, which is determined by the task execution delay $d _ {i}$ , the number of task drop $\eta _ {i}$, the task queuing delay $\rho _ {i}$ and the payment $\phi _ { i }$. Besides, it is worth mentioning that this optimization policy can be personalized according to the target. For instance, if low delay is the most important indicator in a system, the weighting of the task execution delay $d _ {i}$ and task queuing delay $\rho _ {i}$ can be adjusted to change the proportion of delay in the whole utility.

\subsection{Complexity Analysis}
\label{sec:complexity analysis}
Here, we can first consider a special case of this problem. Suppose that at epoch $t$, there are $N^{\prime}(q_{\mathrm{max}}^{\mathrm{t}} \geq N^{\prime}>0)$ tasks in the task queue and $M^{\prime}(q_{\mathrm{max}}^{\mathrm{e}} \geq M^{\prime}>0)$ energy units in the energy queue. Besides, no energy unit and task are generated after epoch $t$. In addition, each task is executed on the local device and executed at the maximum CPU frequency $f_{\mathrm{max}}^{c}$. Therefore, the energy unit $e_k^{\prime}$ required to perform task $k$ and the utility $ u_k^{\prime}$ obtained by completing task $k$ will be two certain values. Furthermore, $d_k^{\prime} \in \{0,1\}$ is defined as the task execution indicator, that is, $ d_k^{\prime}=0$ means that the task $k$ is not executed, and $d_k^{\prime}=1$ means that the task $k$ is executed. In this case, the problem is changed to:
\begin{equation}
	\label{complexity1}
	\begin{aligned}
		\max \sum_{k \in \ N^{\prime}} d_k^{\prime} \cdot u_k^{\prime}
	\end{aligned}~,
\end{equation}
subject to
\begin{equation}
	\label{complexity2}
	\begin{aligned}
		\sum_{k \in \ N^{\prime}} d_k^{\prime} \cdot e_k^{\prime} \leq M^{\prime}.
	\end{aligned}~
\end{equation}

For this situation, we can regard the energy units $ M^{\prime}$ as the capacity of a knapsack, $N^{\prime}$ tasks as the items, and the energy unit $e_k^{\prime}$ required for each task and the utility $ u_k^{\prime}$ obtained by completing a task are regarded as the weight and value of items respectively. Then the special case can be regarded as 0/1 knapsack problem. Since \cite{johnson1979computers} has proven that the 0/1 knapsack problem is NP-hard. Therefore, the special case is also NP-hard. According to \cite{garey1978strong}, since the problem in the special case is NP-hard, the problem is also NP-hard in non-special case.

\subsection{Reasons for Using Federated Learning in Edge Computing}

In the above, we have introduced the computation offloading problem. Fortunately, DRL can deal with this kind of problem well, thus we use Double Deep Q Learning (DDQN) \cite{Mnih2015, Hasselt20016} to maximize the long-term utility. In addition, we also use the FL in our policy, and then we will explain why FL is used.

Although DRL can make decisions efficiently, it consumes a lot of computing resources. Therefore, how to provide the computing resources needs to be considered. On one hand, if the DRL agent is trained on the EN, it will bring about three disadvantages: 1) it will cause a large amount of data to be transmitted between the IoT device and EN, thereby increasing the transmission pressure of the wireless channel; 2) the transmitted data may contain private information, which is not conducive to the privacy protection of the data; 3) although the privacy information in the data can be removed in some ways, this will destroy the integrity of the data and affect the training effect.

On the other hand, if the DRL agent is trained on the IoT device individually, two deficiencies remain: 1) this will take too long to train each DRL agent from the beginning; 2) if each DRL agent will be trained independently, it will cause more energy consumption.

\begin{figure}[htpb]
	\centering
	\includegraphics[width=0.8\textwidth]{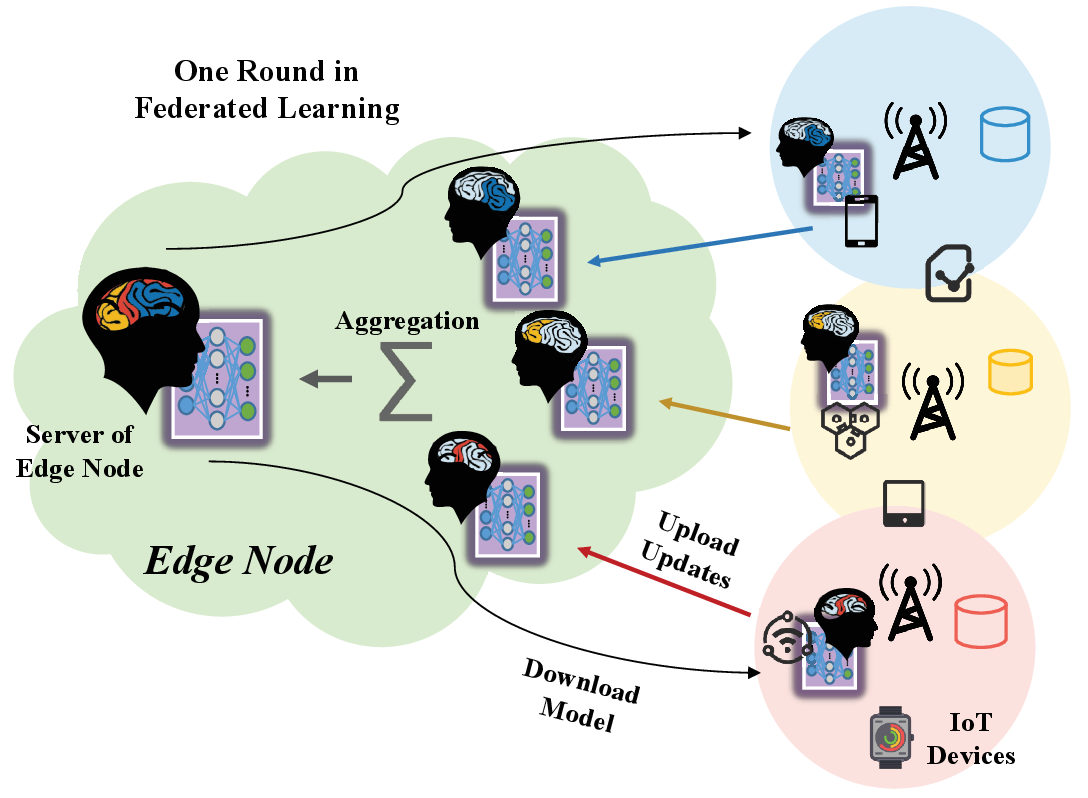}
	\caption{Federated Learning-based DRL Training.}
	\label{fig:fl_drl_training}
\end{figure}

Therefore, the DRL will be trained in a distributed manner as shown in Fig. \ref{fig:fl_drl_training}. However, due to network limitations and the challenge of protecting data privacy, most distributed deep learning technologies \cite{mcmahan2016communication} are not feasible. For the above reasons, FL is introduced into the proposed policy for distributed training DRL agents.

\subsection{Federated Learning-based DRL Training about Computation Offloading}

Because each IoT device needs to make decisions based on its own network environment in computation offloading and it will also have an impact on the network environment. Therefore, we need EN to coordinate various IoT devices to optimize the overall network environment.

As shown in Algorithm 1, some IoT devices will be randomly selected during each iteration to perform the following operations: \textit{1)} download DRL agent parameters from EN and load them; \textit{2)} training DRL agent with the data obtained by itself; \textit{3)} upload update parameters to EN for model aggregation. In the envisioned IoT system, the IoT device needs to train the DRL agent based on the local data acquired by itself, including the number of unfinished tasks, the remaining energy units, and their own network connection status, etc. Therefore, the IoT device does not need to upload these local data, just upload the update parameters to the EN for aggregation, and then download the aggregated model parameters from EN. In addition, some IoT devices with insufficient training data can share the training of the DRL agent.

\begin{algorithm}[!http]
	\caption{Federated Learning-based Policy Training}\label{algo1}
	
	\textbf{Initialization:}\\
	 Init the set of ENs $\mathcal { N }$, the set of IoT devics $\mathcal { D }$;\\
	 Init \{$\theta_n$, $\theta_d$, $C_{d}$|$n \in \mathcal { N }$, $d \in \mathcal { D }$\}

	\textbf{Iteration:}\\
	\For{$t = 1,2,...,T$}{
		Randomly select m IoT devices $\mathcal { D }^t \subseteq \mathcal { D }$;\\
		\For{each $d \in \mathcal { D }^t$}  
		{  
			Gets weights $\theta_n^{t-1}$ of the associated EN;\\  
			Update local weights $\theta_d^{t-1} \Leftarrow \theta_n^{t-1}$;\\
			Get local data $\mathcal{D}_{d}^t$;\\
			Get weights and training times ($\theta_d^{t}, C_{d}^{t}) \Leftarrow$ Train($\theta_d^{t-1}$, $\mathcal{D}_{d}^t$);\\
			Upload $\theta_d^{t}$ and $C_{d}^{t}$ to EN;\\
		}  
	    \For{each $e \in \mathcal { N }$}  
	    {  
	    	Receive updates from $\mathcal { D }_n^t \subseteq \mathcal { D }$;\\  
	    	$ \theta_n ^ { t  } \Leftarrow \sum\nolimits_{d \in \mathcal { D }^t_n} { (C_d^t / \sum\nolimits_{d \in \mathcal { D }^t_n} {C_d^t}) \cdot \theta ^ {t } _ {d} }$ ;\\
	    }  
	 }
\end{algorithm}

\subsection{Theoretical Analysis of Federated Learning-based Policy}

Recalling the federal learning-based policy training above, for each agent running on the IoT device, the network state $X_{i}=\left(q_{i}^{\mathrm{t}}, q_{i}^{\mathrm{e}}, s_{i}, \boldsymbol{g}_{i}\right) \in \mathcal{X}$ will be obtained first, then the action pair $\left(c_{i}, e_{i}\right)$will be selected according to the policy $\mathbf{\Phi}$, after that, a new network state $X_{i+1}$ and a reward $r_i$ will be generated. The process can be described as a Markov decision process.

The more important it is to be closer to the current reward when calculating reward feedback. Therefore, the method of discounted future reward is adopted
\begin{equation}
	\begin{aligned}
		R_{t}=r_{t}+\gamma r_{t+1}+\gamma^{2} r_{t+2}+\ldots+\gamma^{n-t} r_{n}=r_{t}+\gamma R_{t+1}~,
	\end{aligned}
	\label{reward}
\end{equation}
where $\gamma $ = 0.9 is the discount factor. After that, the basic form of the bellman equation can be obtained by combining equation(\ref{reward})
\begin{equation}
	v(\mathcal{X}_{t})=\mathbf{E}\left[r_{t}+\gamma v\left(\mathcal{X}_{t+1}\right) | \mathcal{X}_{i}=\mathcal{X}_{t}\right]~,
	\label{bellman}
\end{equation}

Similar, Q-function is introduced to describe the value of different actions in a certain state. Therefore, $Q(\mathcal{X}^{j}, (c^{j}, e^{j}))$ defined to represent the value of the action pair $(c^{j},e^{j})$ when it is in state $\mathcal{X}^{j}$.
\begin{equation}
	Q(\mathcal{X}^{j+1}, (c^{j+1}, e^{j+1}))=\mathbf{E}\left[r+\lambda Q\left(\mathcal{X}^{j}, (c^{j}, e^{j})\right) \right]~
	\label{q_fun}
\end{equation}

However, because the greedy method is used to calculate the Q-function, the calculation of the Q-function may be too close to an earlier calculated local optimal Q-function, resulting in a large deviation, also called the overestimation. To solve overestimation, DDQN decouples the selection of actions and calculates the Q-function~\cite{Hasselt20016}. First, it will find the action pair $(c^{\mathrm{max}}, e^{\mathrm{max}})$ corresponding to the maximum Q-function through Current Network, and then calculate Q-function through Target Network. In addition, for the theoretical proof, the deviation caused by the use of neural network approximation is neglected here, namely

\begin{equation}
	\begin{aligned}
		Q(\mathcal{X}^{\prime},(c^{\prime}, e^{\prime}))=Q(\mathcal{X}^{j},(c^{j}, e^{j}))+\alpha\left( u \left( \mathcal{X} _ { j } , \Phi \left( \mathcal{X} _ { j } \right) \right)+\gamma \cdot Q\left(\mathcal{X}^{\prime},\left((c^{\mathrm{max}}, e^{\mathrm{max}}\right)\right)-Q(\mathcal{X}^{j},(c^{j}, e^{j}))\right)~,
	\end{aligned}
	\label{ddqn_q_fun}
\end{equation}
where $\alpha$ = 0.005 is the learning rate and $\mathcal{X}^{\prime}$, $c^{\prime}$ and $e^{\prime}$ represent $\mathcal{X}^{j}$, $c^{j}$ and $e^{j}$ of the next cycle, respectively. Based on the above reasoning, the convergence of the algorithm can be further analyzed, namely:
\begin{theorem}
	\label{theorem1}
	If the following conditions are met, each agent running on the IoT device in the Algorithm 1 will converge to the optimal $Q\left(\mathcal{X}^{*},(c^{*}, e^{*})\right)$ with probability one (w.p.1) when it is updated according to equation (\ref{ddqn_q_fun}).
	\begin{itemize}
		\item [1)] 
		The state and action spaces are finite;       
		\item [2)]
		$\sum_{t=0}^{+\infty} \alpha=\infty$, $\sum_{t=0}^{+\infty}\left(\alpha\right)^{2}<\infty$;
		\item [3)]
		\textbf{Var}\{$u \left( \mathcal{X} _ { j } , \Phi \left( \mathcal{X} _ { j } \right) \right)$\} is bounded.
	\end{itemize}
\end{theorem}    
\begin{proof}
	In the proof, a result of stochastic approximation given in \cite{jaakkola1994convergence} is first used.
	
	\begin{lemma}
		\label{lemma1}
		A random iterative process $\triangle^{j+1}(x)$, which is defined as
		\begin{equation}
			\Delta^{j+1}(x)=\left(1-\alpha^{j}(x)\right) \triangle^{j}(x)+\beta^{j}(x) \Psi^{j}(x)~,
			\label{lemequ}
		\end{equation}
		
		converges to zero with probability one (w.p.1) if and only if the following conditions are satisfied.
		\begin{itemize}
			\item [1)] 
			The state space is finite;       
			\item [2)]
			$\sum_{j=0}^{+\infty} \alpha^{j}=\infty, \sum_{j=0}^{+\infty}\left(\alpha^{j}\right)^{2}<\infty, \sum_{j=0}^{+\infty} \beta^{j}=\infty, \sum_{j=0}^{+\infty}\left(\beta^{j}\right)^{2}<\infty$, and 
			
			$E\left\{\beta^{j}(x) | \Lambda^{j}\right\} \leq E\left\{\alpha^{j}(x) | \Lambda^{j}\right\}$ uniformly w.p. 1;
			\item [3)]
			$\left\|\mathbf{E}\left\{\Psi^{j}(x) | \Lambda^{j}\right\}\right\|_{W} \leq \varrho\left\|\Delta^{j}\right\|_{W},$ where $\varrho \in(0,1)$
			\item [4)]
			$\mathbf{Var}\left\{\Psi^{j}(x) | \Lambda^{j}\right\} \leq C\left(1+\left\|\Delta^{j}\right\|_{W}\right)^{2},$ where $C>0$ is a constant.
		\end{itemize}
		
		Note that $\Lambda^{j}=\left\{\triangle^{j}, \Delta^{j-1}, \cdots, \Psi^{j-1}, \cdots, \alpha^{j-1}, \cdots, \beta^{j-1}\right\}$denotes the past at time slot $j$. $\|\cdot\|+{W}$ denotes some weighted maximum norm.
	\end{lemma}
	
	Subsequent proofs can be made around lemma \ref{lemma1}, including the transformation of the form into the equation (\ref{lemequ}) and the four conditions in the lemma \ref{lemma1}.
	
	\textbf{a.} \textbf{Transformation equation form.}
	First of all, we can rewrite equation (\ref{ddqn_q_fun}) to the following
	\begin{equation}
		\begin{aligned}
			Q(\mathcal{X}^{\prime},(c^{\prime}, e^{\prime}))=(1-\alpha) \cdot Q(\mathcal{X}^{j},(c^{j}, e^{j}))+\alpha \cdot \left( u \left( \mathcal{X} _ { j } , \Phi \left( \mathcal{X} _ { j } \right) \right)+\gamma \cdot Q\left(\mathcal{X}^{\prime},\left(c^{\mathrm{max}}, e^{\mathrm{max}}\right)\right)\right)~.
		\end{aligned}
		\label{ddqn_q_fun2}
	\end{equation}
	
	By subtracting the optimal $Q\left(\mathcal{X}^{*},(c^{*}, e^{*})\right)$  from both sides of equation (\ref{ddqn_q_fun2}), it can rewrited to the following form
	\begin{equation}
		\triangle^{j}\left(\mathcal{X},(c, e)\right)=\left(1-\alpha^{j}\right) \triangle^{j}\left(\mathcal{X},(c, e)\right)+\alpha^{j} \Psi^{j}\left(\mathcal{X},(c, e)\right)~,
		\label{b3}
	\end{equation}
	where
	\begin{equation}
		\triangle^{j}\left(\mathcal{X},(c, e)\right)=Q(\mathcal{X}^{\prime},(c^{\prime}, e^{\prime}))-Q\left(\mathcal{X}^{*},(c^{*}, e^{*})\right)~,
		\label{b4}
	\end{equation}	
	\begin{equation}
		\begin{aligned}
			\Psi^{j}\left(\mathcal{X},(c, e)\right)=u \left( \mathcal{X} _ { j } , \Phi \left( \mathcal{X} _ { j } \right) \right)+\gamma \cdot Q\left(\mathcal{X}^{\prime},\left(c^{\mathrm{max}}, e^{\mathrm{max}}\right) -Q\left(\mathcal{X}^{*},(c^{*}, e^{*})\right)\right)~.
		\end{aligned}
		\label{b5}
	\end{equation}
	
	Therefore, equation (\ref{b3}) can be seen as equation (\ref{lemequ}) in lemma \ref{lemma1} with $\alpha^{j}(x) = \beta^{j}(x)$ and then the theorem \ref{theorem1} will be proved by proving that all four conditions of lemma \ref{lemma1} have been met.  
	
	\textbf{b.} \textbf{Condition 1) in lemma \ref{lemma1}.}
	For condition 1) in lemma \ref{lemma1}, it can be proved by condition 1) in theorem \ref{theorem1}.
	
	\textbf{c.} \textbf{Condition 2) in lemma \ref{lemma1}.}
	The condition 2) in lemma \ref{lemma1} also can be easily proved. Since $\alpha$ and $\beta$ in the equation (\ref{lemequ}) correspond to the learning rate $\alpha \in (0,1)$ in the equation (\ref{b3}), the condition is satisfied.
	
	\textbf{d.} \textbf{Condition 3) in lemma \ref{lemma1}.}
	To prove condition 3) in lemma \ref{lemma1}, the concept of contraction mapping is introduced and then prove that $\Psi^{j}\left(\mathcal{X},(c, e)\right)$ is a contraction mapping.
	\begin{definition}
		For a map $\textbf{H}:\chi \rightarrow \chi$ and any $x_{1},x_{2} \in \chi$, if there exists a constant $\delta \in (0,1)$ satisfies the following equation, then the map $\textbf{H}$ is a contraction mapping.
		\begin{equation}
			\left\|\mathbf{H} x_{1}-\mathbf{H} x_{2}\right\| \leq \delta\left\|x_{1}-x_{2}\right\|~
			\label{b6}
		\end{equation}		
	\end{definition}
	
	Combined with equation (\ref{q_fun}), the optimal $Q\left(\mathcal{X}^{*},(c^{*}, e^{*})\right)$ can be expressed as
	\begin{equation}
		Q\left(\mathcal{X}^{*},(c^{*}, e^{*})\right)=\mathbf{E}\left[u \left( \mathcal{X} _ { * } , \Phi \left( \mathcal{X} _ { * } \right) \right)+\lambda Q\left(\mathcal{X}^{\prime}, (c^{\mathrm{max}}, e^{\mathrm{max}})\right) \right]~.
		\label{b8}
	\end{equation}
	
	Then further define the mapping \textbf{H} as
	\begin{equation}
		\textbf{H}_{q\left(\mathcal{X}^{\prime},(c^{\prime}, e^{\prime})\right)}=\mathbf{E}\left[u \left( \mathcal{X} _ { j } , \Phi \left( \mathcal{X} _ { j } \right) \right)+\lambda q\left(\mathcal{X}^{\prime}, (c^{\mathrm{max}}, e^{\mathrm{max}})\right) \right]~.
		\label{b9}
	\end{equation}
	
	After that, combined with the properties of absolute value inequalities and the definition of infinity norm, the following calculations can prove that $\textbf{H}$ is a contraction mapping.
	\begin{equation}
		\begin{aligned}
			\left\|\textbf{H}_{q_1\left(\mathcal{X}^{\prime},(c^{\prime}, e^{\prime})\right)}-\textbf{H}_{q_2\left(\mathcal{X}^{\prime},(c^{\prime}, e^{\prime})\right)}\right\|_{\infty}
			=&\mathbf{E}\left[\lambda q_1\left(\mathcal{X}^{\prime}, (c^{\mathrm{max}}, e^{\mathrm{max}})\right)-\lambda q_2\left(\mathcal{X}^{\prime}, (c^{\mathrm{max}}, e^{\mathrm{max}})\right) \right] \\
			\leq& \mathbf{E}\left[\lambda \left|  q_1\left(\mathcal{X}^{\prime}, (c^{\mathrm{max}}, e^{\mathrm{max}})\right)- q_2\left(\mathcal{X}^{\prime}, (c^{\mathrm{max}}, e^{\mathrm{max}})\right) \right| \right] \\
			\leq& \mathbf{E}\left[\lambda \mathop{\arg\max} _{(c^{\prime}, e^{\prime})} \left|q_1\left(\mathcal{X}^{\prime}, (c^{\prime}, e^{\prime})\right)-q_2\left(\mathcal{X}^{\prime}, (c^{\prime}, e^{\prime})\right)\right| \right] \\
			=& ~\lambda\left\|q_1\left(\mathcal{X}^{\prime},(c^{\prime}, e^{\prime})\right)-q_2\left(\mathcal{X}^{\prime},(c^{\prime}, e^{\prime})\right)\right\|_{\infty}
		\end{aligned}
		\label{b10}
	\end{equation}
	
	According to equations (\ref{b9}) and (\ref{b5}), $E\left\{\Psi^{j}(x)\right\}$ can be expressed as
	\begin{equation}
		\begin{aligned}
			\mathbf{E}\left\{\Psi^{j}(x)\right\} &= \mathbf{E}\left\{u \left( \mathcal{X} _ { j } , \Phi \left( \mathcal{X} _ { j } \right) \right)+\gamma \cdot Q\left(\mathcal{X}^{\prime},\left(c^{\mathrm{max}}, e^{\mathrm{max}}\right)\right) -Q\left(\mathcal{X}^{*},(c^{*}, e^{*})\right)\right\}\\
			&=\textbf{H}_{Q\left(\mathcal{X}^{\prime},(c^{\prime}, e^{\prime})\right)}-Q\left(\mathcal{X}^{\prime},(c^{\prime}, e^{\prime})\right)\\
			&=\textbf{H}_{Q\left(\mathcal{X}^{\prime},(c^{\prime}, e^{\prime})\right)}-\textbf{H}_{Q\left(\mathcal{X}^{*},(c^{*}, e^{*})\right)}.
		\end{aligned}
		\label{b11}
	\end{equation}
	
	The last step in equation (\ref{b11}) is because $Q\left(\mathcal{X}^{\prime},(c^{\prime}, e^{\prime})\right)$ is a constant. Finally, we can prove that Condition (3) in Lemma 4.2 satisfies through Equation (39). In the calculation process of Equation (39), the first step can be derived based on Equation (38), the second step can be derived based on Equation (34), and the last step can be derived based on Equation (32).
	\begin{equation}
		\begin{aligned}
			\left\|\mathbf{E}\left\{\Psi^{j}(x)\right\}\right\|_{\infty} &=\left\|\textbf{H}_{Q\left(\mathcal{X}^{\prime},(c^{\prime}, e^{\prime})\right)}-\textbf{H}_{Q\left(\mathcal{X}^{*},(c^{*}, e^{*})\right)}\right\|_{\infty}\\
			&\leq \delta \left\| Q\left(\mathcal{X}^{\prime},(c^{\prime}, e^{\prime})\right)- Q\left(\mathcal{X}^{*},(c^{*}, e^{*})\right) \right\|_{\infty} \\
			&= \delta \left\|  \triangle^{j}\left(\mathcal{X},(c, e)\right)  \right\|_{\infty}
		\end{aligned}
		\label{b12}
	\end{equation}
	
	\textbf{e.} \textbf{Condition 4) in lemma \ref{lemma1}.}
	Next is the proof of the condition 4) in lemma \ref{lemma1}. The following calculation uses equation (\ref{b5}), equation (\ref{b9}) and the properties of variance.
	\begin{equation}
		\begin{aligned}
			\mathbf{Var}\left\{\Psi^{j}\left(\mathcal{X},(c, e)\right)\right\}
			=&\mathbf{E}\left\{{ u\left( \mathcal{X} _ { j } , \Phi \left( \mathcal{X} _ { j } \right) \right)+\gamma \cdot Q\left(\mathcal{X}^{\prime},\left(c^{\mathrm{max}}, e^{\mathrm{max}}\right)\right) -Q\left(\mathcal{X}^{*},(c^{*}, e^{*})\right)}\right.\\
			& \left.{- \textbf{H}_{Q\left(\mathcal{X}^{\prime},(c^{\prime}, e^{\prime})\right)}+Q\left(\mathcal{X}^{*},(c^{*}, e^{*})\right)}\right\} \\
			=&\mathbf{E}\left\{ u\left( \mathcal{X} _ { j } , \Phi \left( \mathcal{X} _ { j } \right) \right)+\gamma \cdot Q\left(\mathcal{X}^{\prime},\left(c^{\mathrm{max}}, e^{\mathrm{max}}\right)\right)  - \textbf{H}_{Q\left(\mathcal{X}^{\prime},(c^{\prime}, e^{\prime})\right)}\right\} \\
			=&\mathbf{E}\left\{ {u\left( \mathcal{X} _ { j } , \Phi \left( \mathcal{X} _ { j } \right) \right)+\gamma \cdot Q\left(\mathcal{X}^{\prime},\left(c^{\mathrm{max}}, e^{\mathrm{max}}\right)\right)  }\right. \\
			&\left.{ - \mathbf{E}\left[ u\left( \mathcal{X} _ { j } , \Phi \left( \mathcal{X} _ { j } \right) \right)+\gamma \cdot Q\left(\mathcal{X}^{\prime},\left(c^{\mathrm{max}}, e^{\mathrm{max}}\right)\right)  \right]}\right\} \\
			=&\mathbf{Var}\left\{ u\left( \mathcal{X} _ { j } , \Phi \left( \mathcal{X} _ { j } \right) \right)+\gamma \cdot Q\left(\mathcal{X}^{\prime},\left(c^{\mathrm{max}}, e^{\mathrm{max}}\right)\right) \right\}\\
			\leq& C(1+ \left\|  \triangle^{j}\left(\mathcal{X},(c, e)\right)  \right\|_{\infty}),
		\end{aligned}
		\label{b13}
	\end{equation}
	where $C$ is a constant and the last step in the above calculation is because of $u\left( \mathcal{X} _ { j } , \Phi \left( \mathcal{X} _ { j } \right) \right)$is bounded and $ Q\left(\mathcal{X}^{\prime},\left(c^{\mathrm{max}}, e^{\mathrm{max}}\right)\right)$ at most linearly.
	
\end{proof}

The above proof process has proved that the four conditions in lemma \ref{lemma1} are satisfied, and based on this, theorem \ref{theorem1} can be proved.

\section{Performance Evaluation}
\label{sec:performance_evaluation}
In this section, we conduct some simulation experiments to explore the impact of various parameters on FL-based DRL training and the performance of FL-based DRL training and centralized DRL training.

\subsection{Experiment Settings}
To evaluate the capability of our proposed method, we simulate an edge computing system. The channel gain of the network connection between the IoT device and EN is defined as 6 levels. Besides, the settings for the DRL agent are shown in table \ref{drl_agent}.

\begin{table*}[htbp]
	\centering  
	\fontsize{6.5}{8}\selectfont  
	\begin{threeparttable}  
		\caption{The settings for the DRL agent}  
		\label{drl_agent}  
		\begin{tabular}{m{6cm}<{\centering}m{2cm}<{\centering}}
			\toprule  
			\multicolumn{1}{c}{\bf Parameters }&\multicolumn{1}{c}{\bf Values}\cr
			\midrule 
			The number of full connected layers&$2$\cr
			The number of neurons&$200$\cr
			Activation function&$\tanh$\cr
			Discount factor&$0.9$\cr
			Replay memory capacity&$5000$\cr
			Update interval of the target $Q$ network&$250$\cr
			\bottomrule  
		\end{tabular}  
	\end{threeparttable}  
\end{table*}
\subsection{Analysis of Exploration Probability}

\begin{figure}[htbp]
	\centering
	\includegraphics[width=0.8\textwidth]{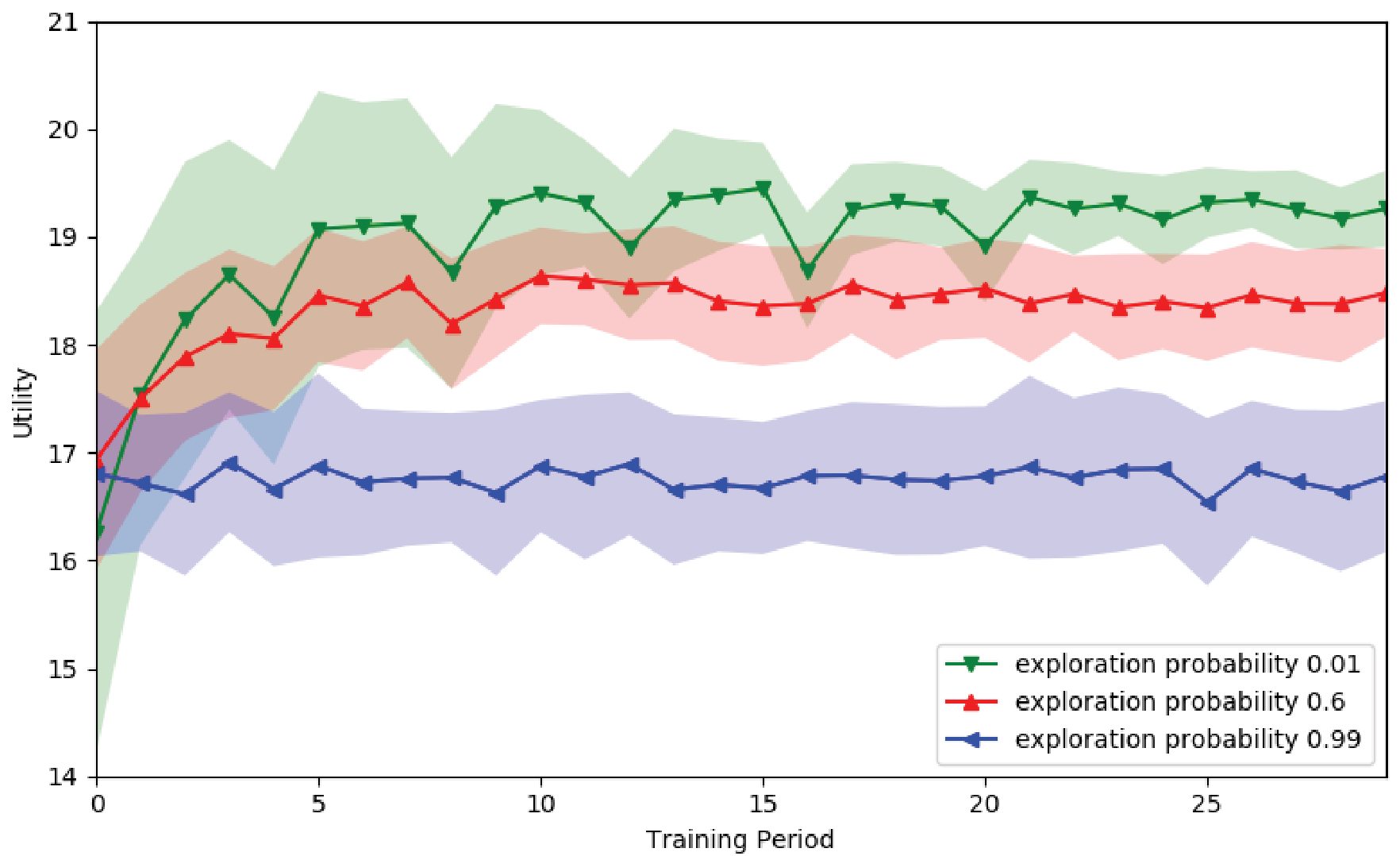}
	\caption{Exploration effect.}
	\label{fig:exploration_effect}
\end{figure}

A strategy called $\epsilon-greedy$ is used in DDQN, which is a similar approach to greedy strategy. Since there is no corresponding $Q$ value for the state-action pair that has not appeared and these combinations will never be tried if the greedy strategy is used directly, so the concept of exploration and utilization is integrated. On the one hand, we must seek to maximize the utility from the information already known, and on the other hand, we must explore what is not known in the environment. The $\epsilon-greedy$ is to add exploration rates based on a greedy strategy. For each decision, there will be a probability to randomly select actions to explore, and there is a probability to use a greedy strategy to maximize the utility. Therefore, we first explored the effect of exploration probability on utility in the experiment.

We selected three values for exploration probability to test and one hundred times of experiments are taken for collecting statics. In Fig. \ref{fig:exploration_effect}, solid lines and the shallow area around them represent mean values and the standard deviation, respectively. It can be found that the average value of utility is the highest when exploration probability = $0.01$, and the standard deviation is the smallest. When exploration probability = $0.6$, the average value of utility will decrease slightly and the standard deviation will increase. When exploration probability = $0.99$, the average value of utility will be greatly reduced and the range of standard deviation will be further increased. It can be found that a higher exploration probability will have a negative impact on performance of IoT devices. This is due to the fact that IoT devices in most cases choose to explore and do not fully utilize the previous training results, resulting in poor performance. Based on the results, accordingly, we choose the final exploration possibility as $0.01$ in the following experiments.

\subsection{Analysis of Task Generate Probability}
When the IoT device is applied, the workload of some IoT devices is heavy, and some IoT devices are relatively idle. Even for an IoT device, the workload will change greatly during different time periods. Therefore, we further explore the effect of workload.
\begin{figure}[htbp]
	\centering
	\includegraphics[width=0.8\textwidth]{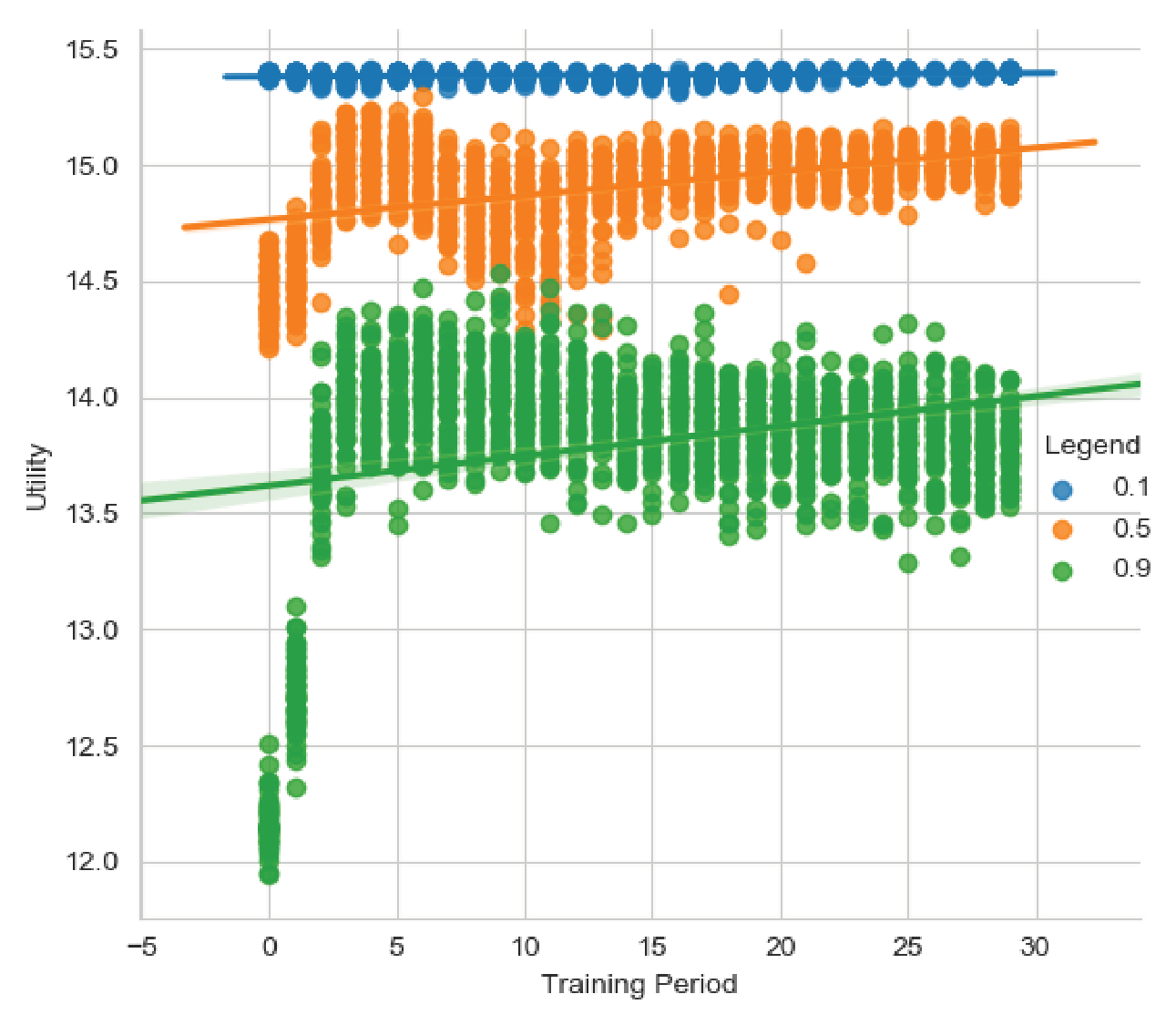}
	\caption{Utility varies with workload.}
	\label{fig:task_1}
\end{figure}

Under different probabilities of task generation, the performance of the proposed algorithm is showed in Figure \ref{fig:task_1}. When the task generates probability = $0.1$, the IoT device can run at a higher utility and the standard deviation of utility is also small, namely, the performance is very good and stable. When the task generates probability = $0.5$, the workload of the IoT device will be relatively heavy. In the initial stage of training, the utility is lower and the standard deviation is large, that is, the performance is poor and unstable. However, as the training progressed, the utility stabilized at a higher level and the standard deviation is greatly reduced. When the task generate probability = $0.9$, the workload of the IoT device is extremely heavy and the overall utility value is very low. It can be found that, when the exploration probability is higher, the performance of the proposed algorithm will be worse and more unstable. In other words, a higher exploration probability will have a negative impact on the performance. This may be due to the limitation of IoT devices' performance, because a high exploration probability will produce a large number of tasks that IoT devices cannot cope with, resulting in high queuing delay and a large number of task failures. In addition, due to the limited energy of IoT devices, heavy tasks will result in low or no energy allocated for each task, thus limiting the performance of IoT devices. However, after a period of training, the IoT device's utility value is improved and stabilized at a relatively high level, which proves the effectiveness of the algorithm under a heavy workload.

\begin{figure}[htbp]
	\centering
	\includegraphics[width=0.8\textwidth]{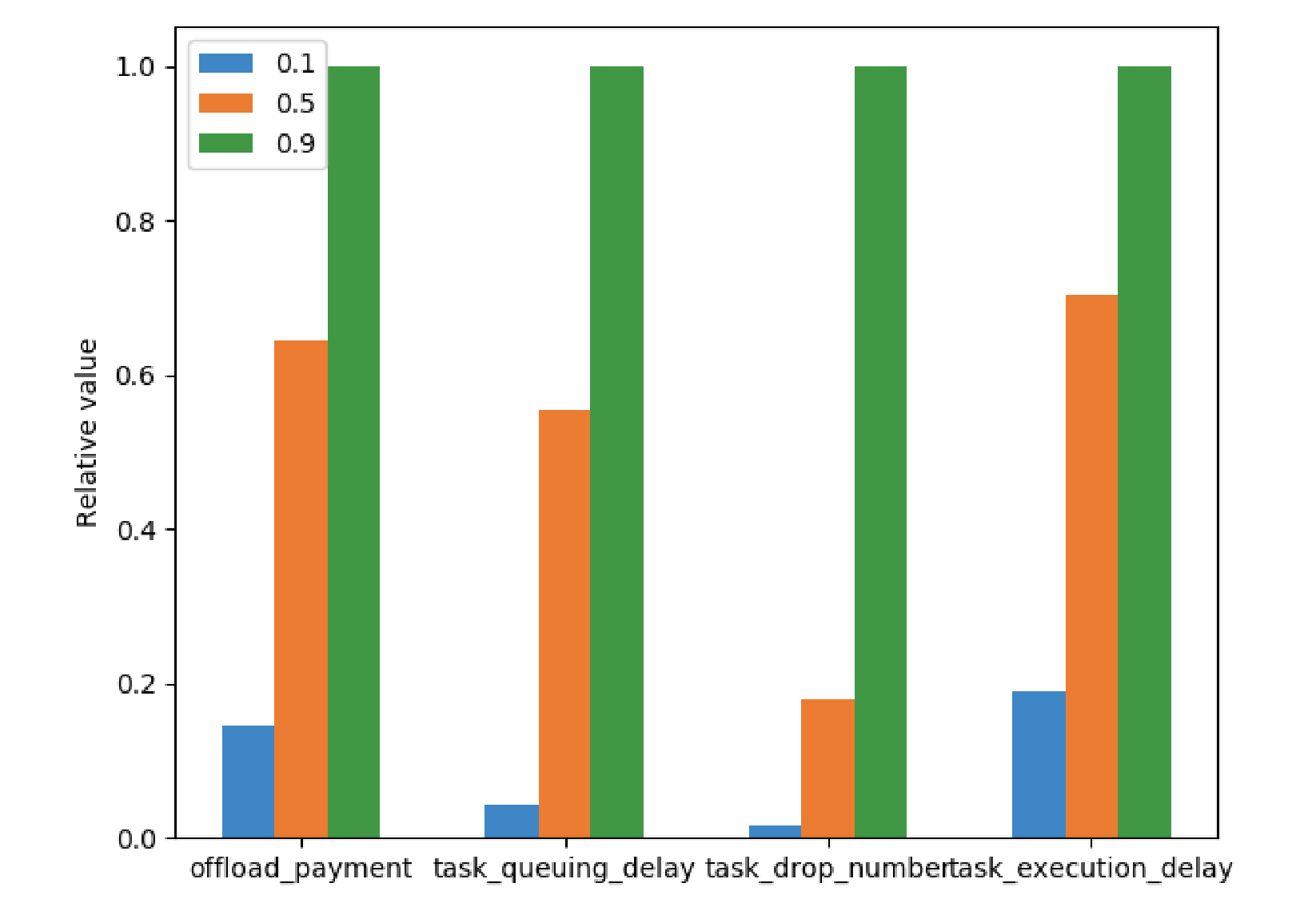}
	\caption{Key parameters that vary with workload.}
	\label{fig:task_2}
\end{figure}

To further explore the performance of our proposed algorithm under different workloads, we selected four key parameters queuing delay, offload payment, task execution delay, and task drop number for comparison. The performance of the algorithm presented under different workloads is shown in Fig. \ref{fig:task_2}. When the task generates probability = $0.1$, the IoT device needs to handle few tasks and is very good at all key parameters. When the task generates probability = $0.5$, the workload of the IoT device is increased, and various key parameters are beginning to increase. However, the increase of the task drop number is small, which means that the IoT device can make most tasks complete normally under the current workload. When the task generate probability = $0.9$, the workload of the IoT device is further increased, which will make the key parameters perform worse.

In summary, Fig. \ref{fig:task_1} and Fig. \ref{fig:task_2} show the performance of one IoT device with varying the workload on itself, which corroborates that our work can adapt to the workload variation and be converged.

\subsection{Analysis of Energy Generate Probability}
Differences in the deployment conditions of IoT devices may cause them to face different probabilities of energy generation. Therefore, in order to test the performance of the algorithm under different energy generation, we conducted the following experiments.

\begin{figure}[htbp]
	\centering
	\includegraphics[width=0.8\textwidth]{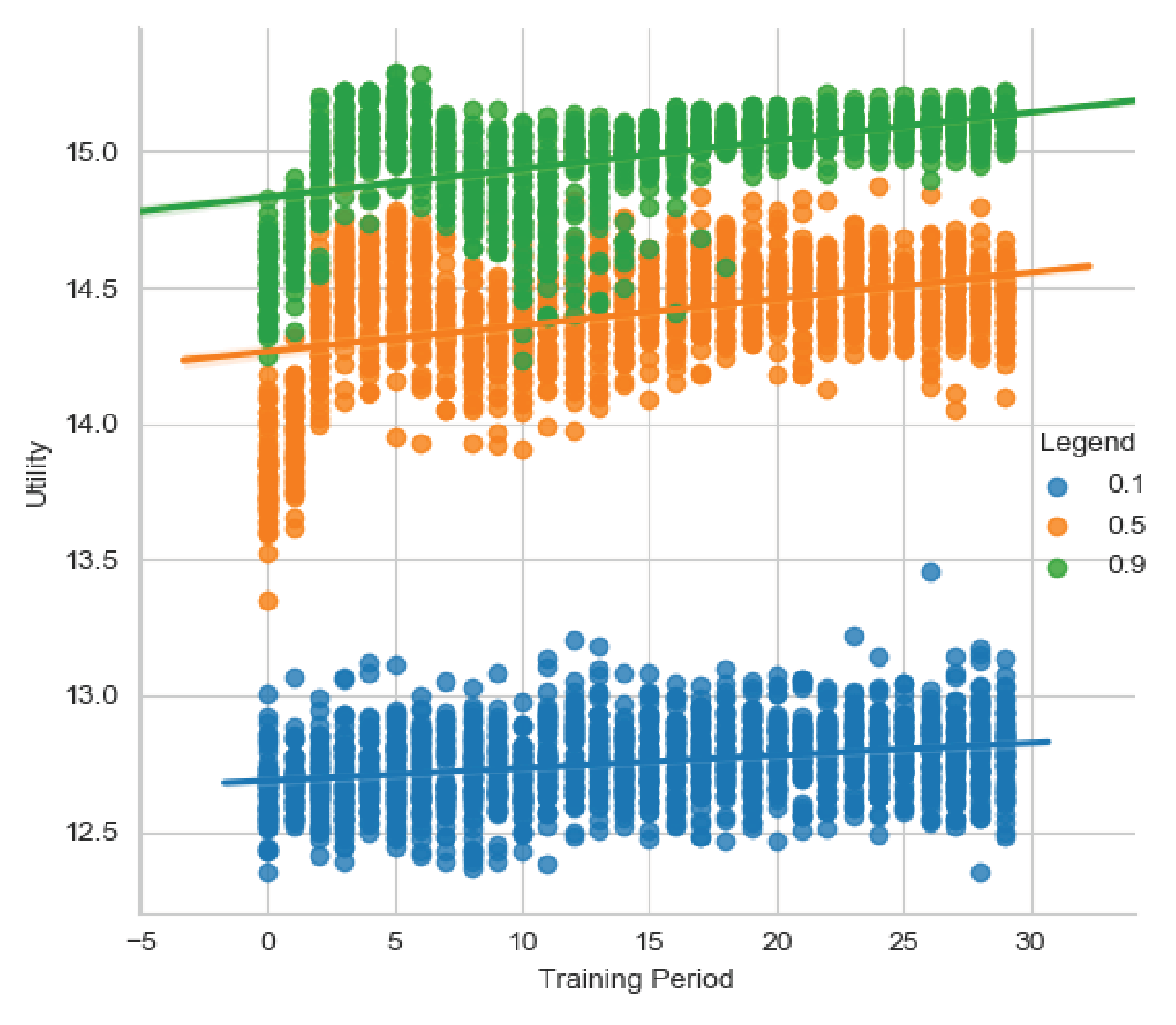}
	\caption{Utility varies with energy generate probability.}
	\label{fig:energy_1}
\end{figure}

Under different probabilities of energy generation, the performance of the proposed algorithm is showed in Figure \ref{fig:energy_1}. When energy generation probability = $0.1$, IoT devices obtain energy at a very slow speed, resulting in energy shortage which leads to a lower utility level. When energy generation probability = $0.5$, the speed of obtaining the energy of IoT equipment increases obviously. Besides, the utility is improved significantly overall and shows a tendency to gradually increase and eventually stabilize. When energy activates probability = $0.9$, the energy of the IoT device is more abundant. The utility is further improved, and the standard deviation is reduced, that is, the IoT device performs better and more stable. It can be found that the higher the energy generate probability is, the more sufficient the energy of IoT equipment is, so higher power can be used to perform tasks and the overall performance is better.

\begin{figure}[htbp]
	\centering
	\includegraphics[width=0.8\textwidth]{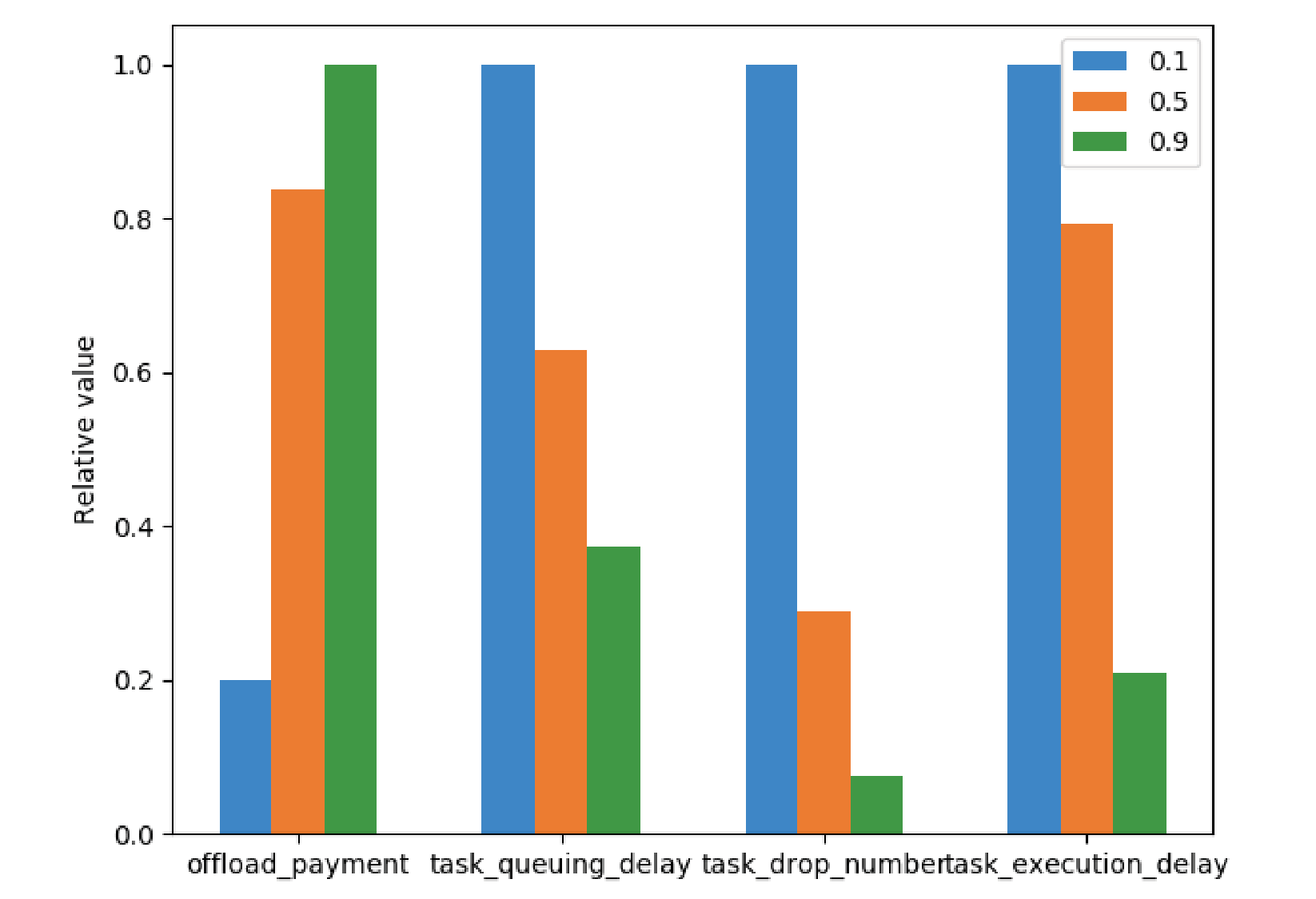}
	\caption{Key parameters that vary with energy generate probability.}
	\label{fig:energy_2}
\end{figure}

To more closely compare the performance of the proposed algorithm in different energy environments, we selected some key parameters for comparison. The algorithm proposed in different energy environments is shown in Fig. \ref{fig:energy_2}. When energy generation probability = $0.1$, the queuing delay, task drop number, and task execution delay is very high which may be due to insufficient energy, resulting in the task can not be completed in time. However, the offload payment at this time is very low, which means that most tasks are executed locally rather than being transferred to EN. When energy generation probability = $0.5$, the IoT device has more energy available, and the offload payment is further improved but other parameters are reduced. This indicates that more tasks are transferred to EN and the task execution is faster. When energy generation probability = $0.9$, the energy of the IoT device is sufficient. In this case, more tasks are transferred to EN for execution and the proposed algorithm performs better. Therefore, it can be seen that the higher the energy generate probability is, the more tasks submitted to EN for running. That is to say, compared with local execution of IoT devices, transferring tasks to EN will consume more energy but obtain more powerful task processing capability.

Therefore, experiments show that the proposed algorithm can adapt to different energy to generate probability and perform convergence.

\subsection{Analysis of the Number of IoT Devices}
When the FL-based DRL training is running, it will be trained using many IoT devices in the area. However, the number of IoT devices used is difficult to determine, so the impact of the number of different IoT devices will be explored next.

\begin{figure}[htbp]
	\centering
	\includegraphics[width=0.8\textwidth]{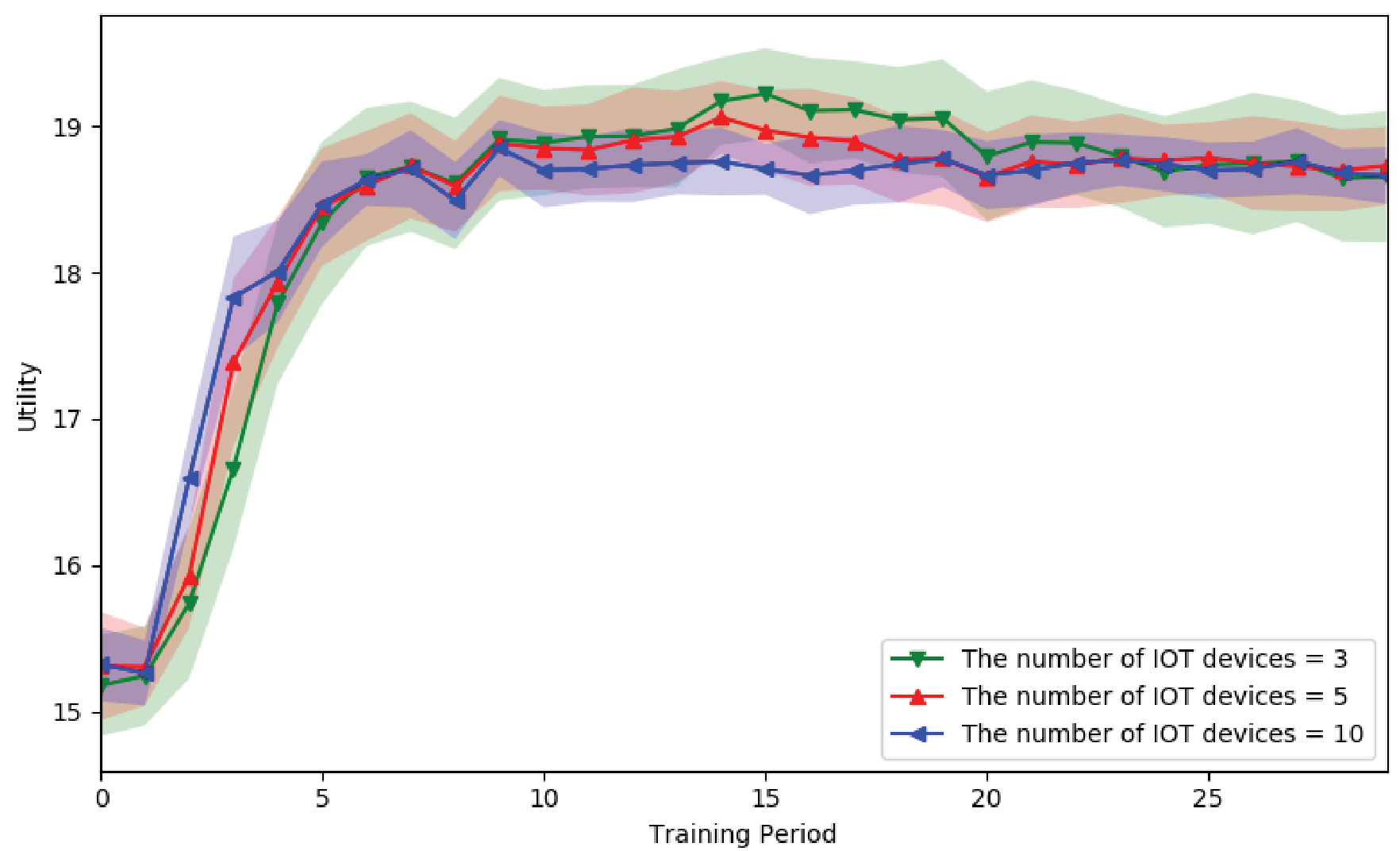}
	\caption{Performance under different numbers of IoT devices.}
	\label{fig:agents}
\end{figure}

The experiment selected different numbers of IoT devices for testing and collected their utilities as a comparison. As shown in Fig. \ref{fig:agents}, all their utilities show an upward trend at the early stage and remain stable at the later stage. However, the number of different IoT devices also has an impact on their utilities. On the one hand, when the number of IoT devices is small, its utility increases relatively slowly in the early days. On the other hand, for different IoT devices, all the utility will be stable at the same level in the later stage, but the number of IoT devices will have an impact on the standard deviation of the utility. When the number of IoT devices is small, the standard deviation of the utility will be large, that is, more IoT devices will make the performance more stable. We think that when the number of IoT devices is large, more environmental states can be learned in the same training cycle, thus contributing to the training progress of IoT devices in the early stage. Therefore, the more IoT devices in the early stage, the performance is better. However, after training convergence, the more IoT devices there are, the worse their performance will be. This may be because different IoT devices face different environmental states, so their model parameters are not applicable to all IoT devices, that is, aggregating model parameters of IoT devices under different environmental states cannot produce optimal results.
\begin{figure*}[htbp]
	\centering
	\subfigure[Comparison of offload payment]
	{\includegraphics[width=0.47\textwidth]{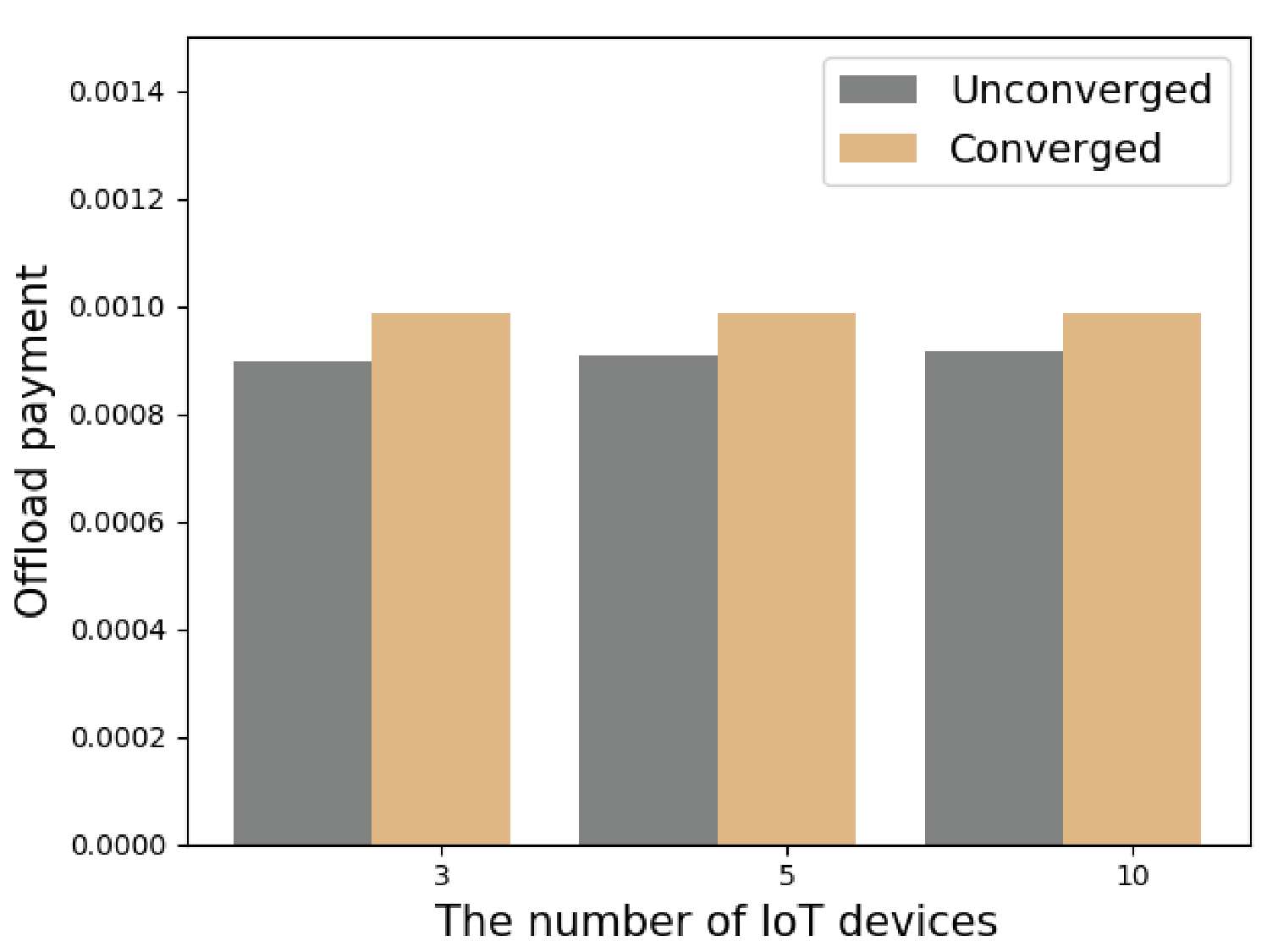}
		\label{histogram1}
	}
	\subfigure[Comparison of the number of task drop]
	{\includegraphics[width=0.47\textwidth]{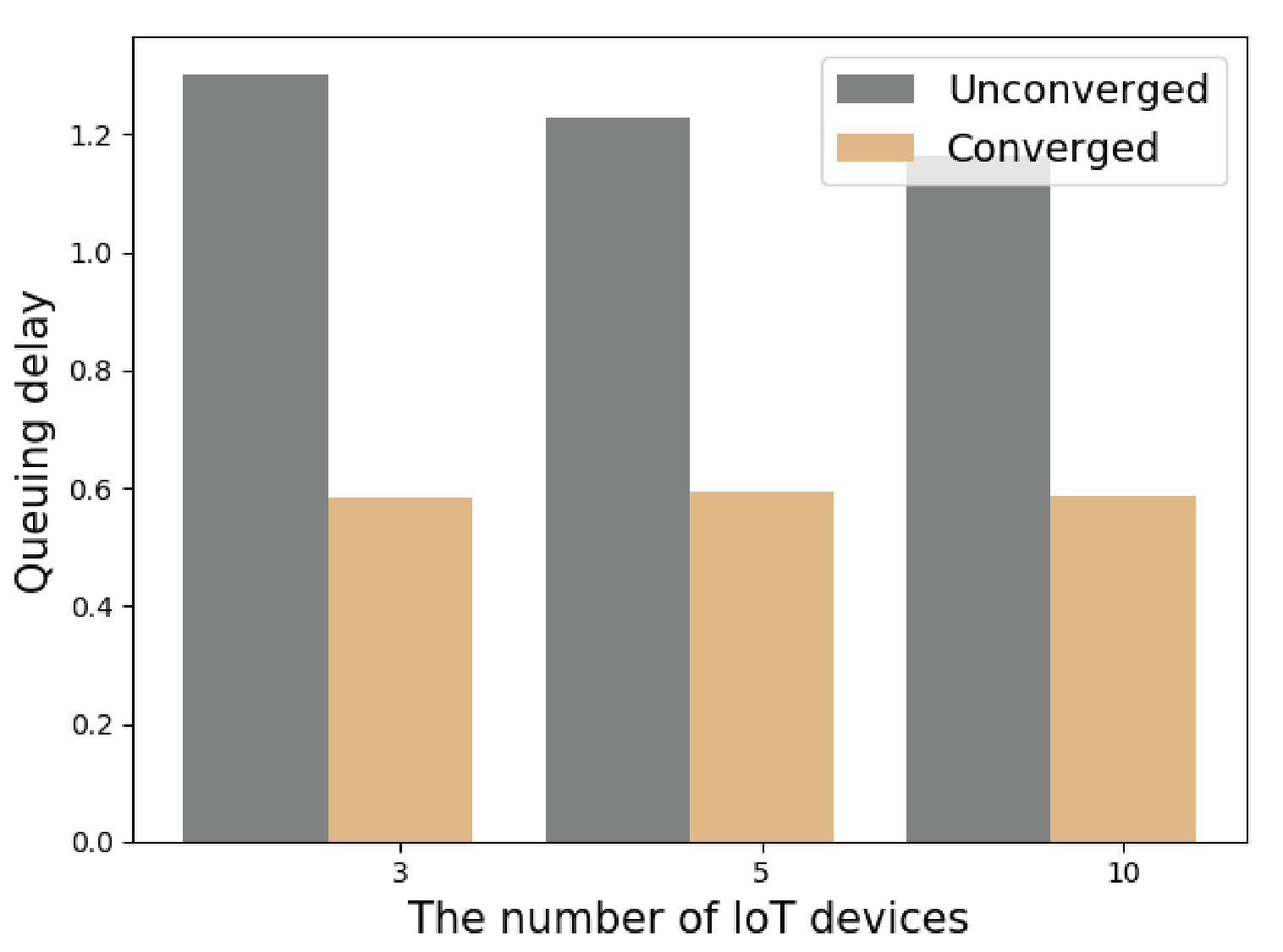}
		\label{histogram2}
	}
	\subfigure[Comparison of task execution delay]
	{\includegraphics[width=0.47\textwidth]{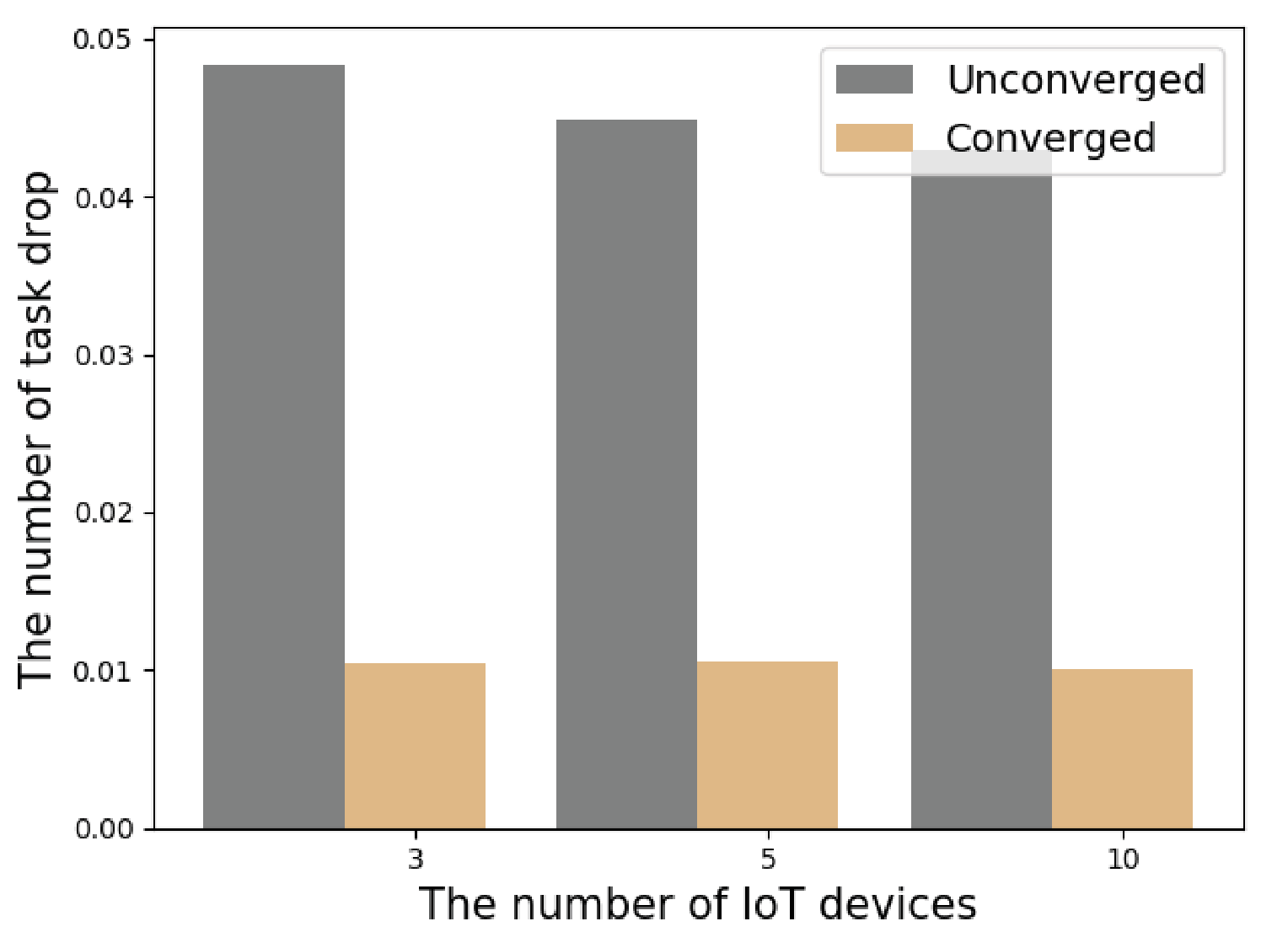}
		\label{histogram3}
	}
	\subfigure[Comparison of queuing delay]
	{\includegraphics[width=0.47\textwidth]{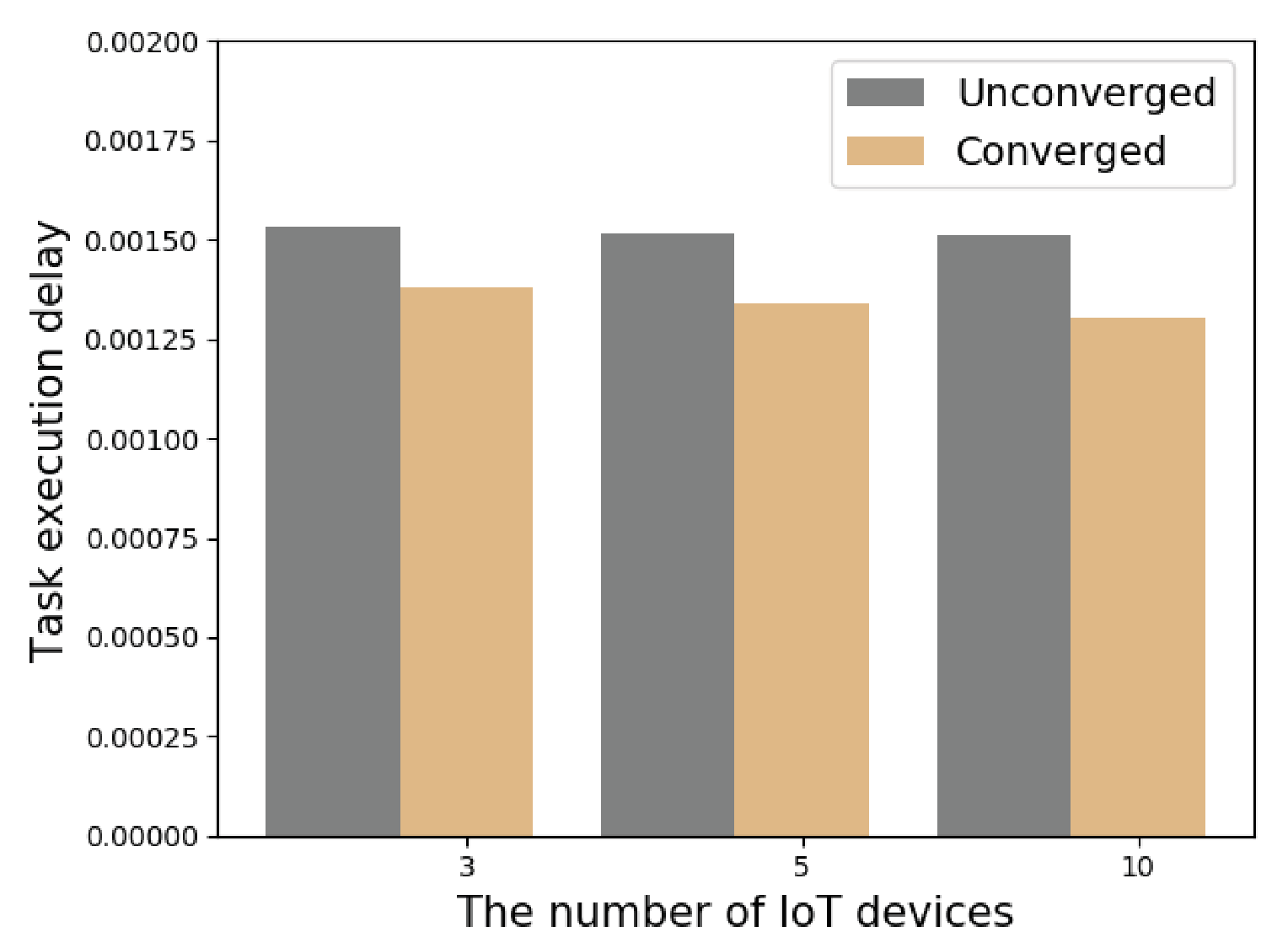}
		\label{histogram4}
	}
	
	\caption{Comparison of parameters before and after convergence for different number of IoT devices.}
	\label{fig:comparison_single}
\end{figure*}

To delve deeper into the impact of the number of different IoT devices, we conducted experiments on more detailed parameters. As shown in Fig. \ref{fig:comparison_single}, when comparing the convergence and non-convergence of a certain IoT device, it can be found that offload payment increases slightly after training convergence, but all other parameters decrease. The reason for this situation may be the increased use of EN after training convergence. Therefore, it costs more to use EN while other parameters are optimized. In addition, for the case of the number of different IoT devices, they are at the same level after the training convergence.

In summary, the more IoT devices, the faster the convergence rate during training, that is, the better the performance in the early stage. However, after training convergence, the performance of different IoT devices is at the same level.

\subsection{Analysis of Energy Consumption}
\begin{figure}[htbp]
	\centering
	\includegraphics[width=0.8\textwidth]{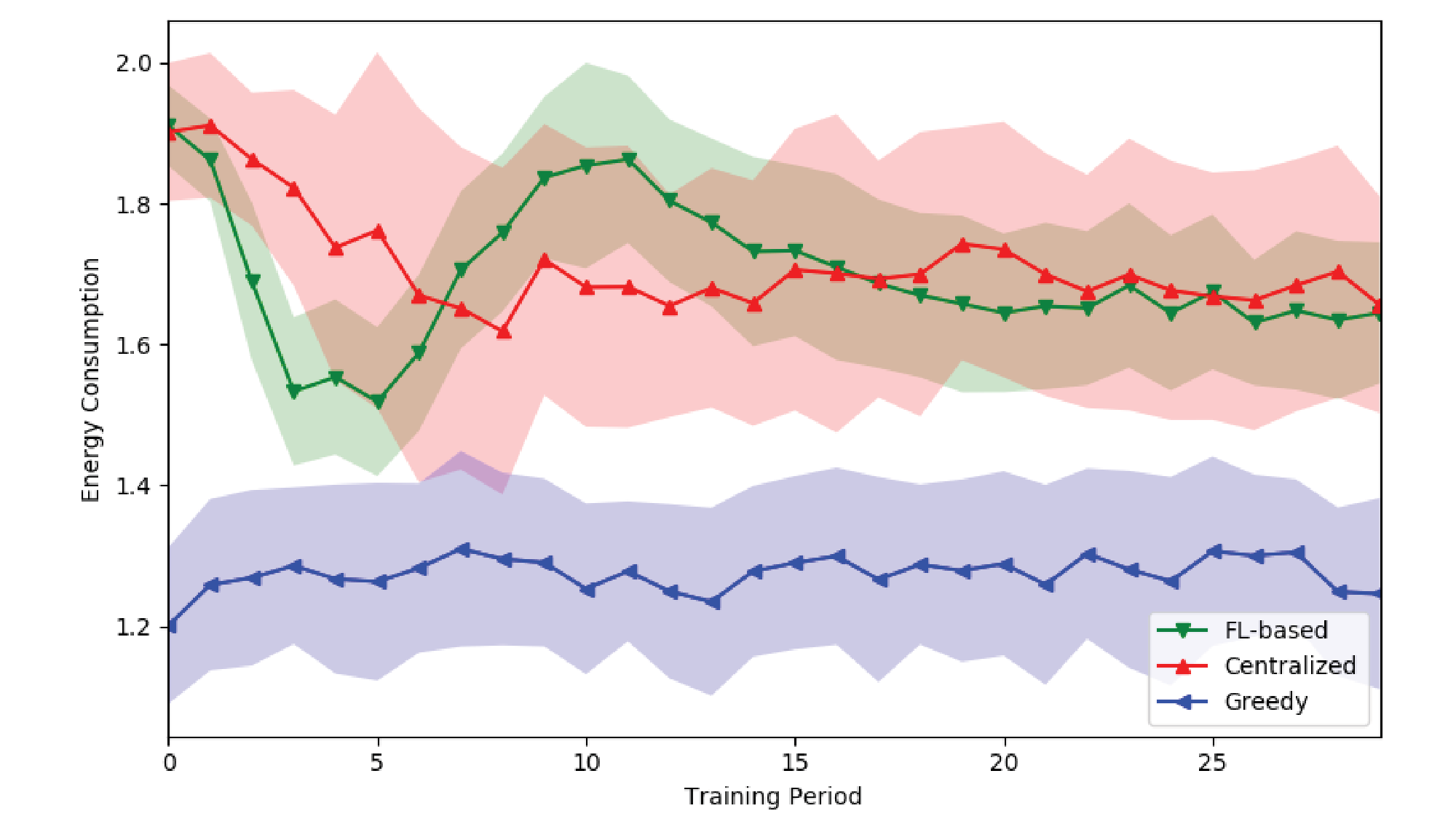}
	\caption{Comparison of energy consumption.}
	\label{fig:energy}
\end{figure}

For IoT devices, many devices use batteries as a power source, so they usually have some limitations in terms of energy. Therefore, energy consumption is a parameter that needs to be considered. As shown in Fig. \ref{fig:energy}, the energy consumption of both FL-based DRL training and centralized DRL training is higher than the greedy strategy. The reason may be that more energy has to be used to reduce the number of task drop and task execution delay. In addition, the average energy consumption of FL-based DRL training is comparable to that of centralized DRL training, but the standard deviation is relatively small.

In summary, our proposed algorithm may result in higher energy consumption, which may be due to higher power for local execution or data transmission. In this case, although the energy consumption may be at a high level, it is improved in terms of time delay, the number of task drop, etc. In addition, it is worth mentioning that the energy consumption of FL-based DRL training and centralized DRL training are at the same level.

\subsection{Analysis of Key Parameters in the System Model}
In order to deeply compare the characteristics of FL-based DRL training and centralized DRL training, some key parameters in the system model are selected for comparison and the first is the task execution delay. Task execution delay is the time from when the task leaves the task queue to completion, including the time of establishing a connection, the time of transferring task data and the time of executing the task. The task execution delay of the FL-based DRL training and centralized DRL training is shown in Fig. \ref{fig:single4}.

It can be seen that the FL-based DRL training has a higher task execution delay than the centralized DRL training task execution delay, and both of them stay stable after falling in the early stage. It is worth noting that the task execution delay of the FL-based DRL training has decreased more in the early stage and reached a steady-state faster, that is, the L-based DRL training has a faster convergence speed. 

In addition, the queuing delay of FL-based DRL training and centralized DRL training is compared. The queuing delay is the time it takes for a task to go from being placed in the task queue to being taken out of the task queue. The queuing delay of FL-based DRL training and centralized DRL training is shown in Fig. \ref{fig:single2}.

\begin{figure}[htbp]    
	\begin{minipage}[t]{0.5\linewidth}    
		\centering    
		\includegraphics[width=1.0\textwidth]{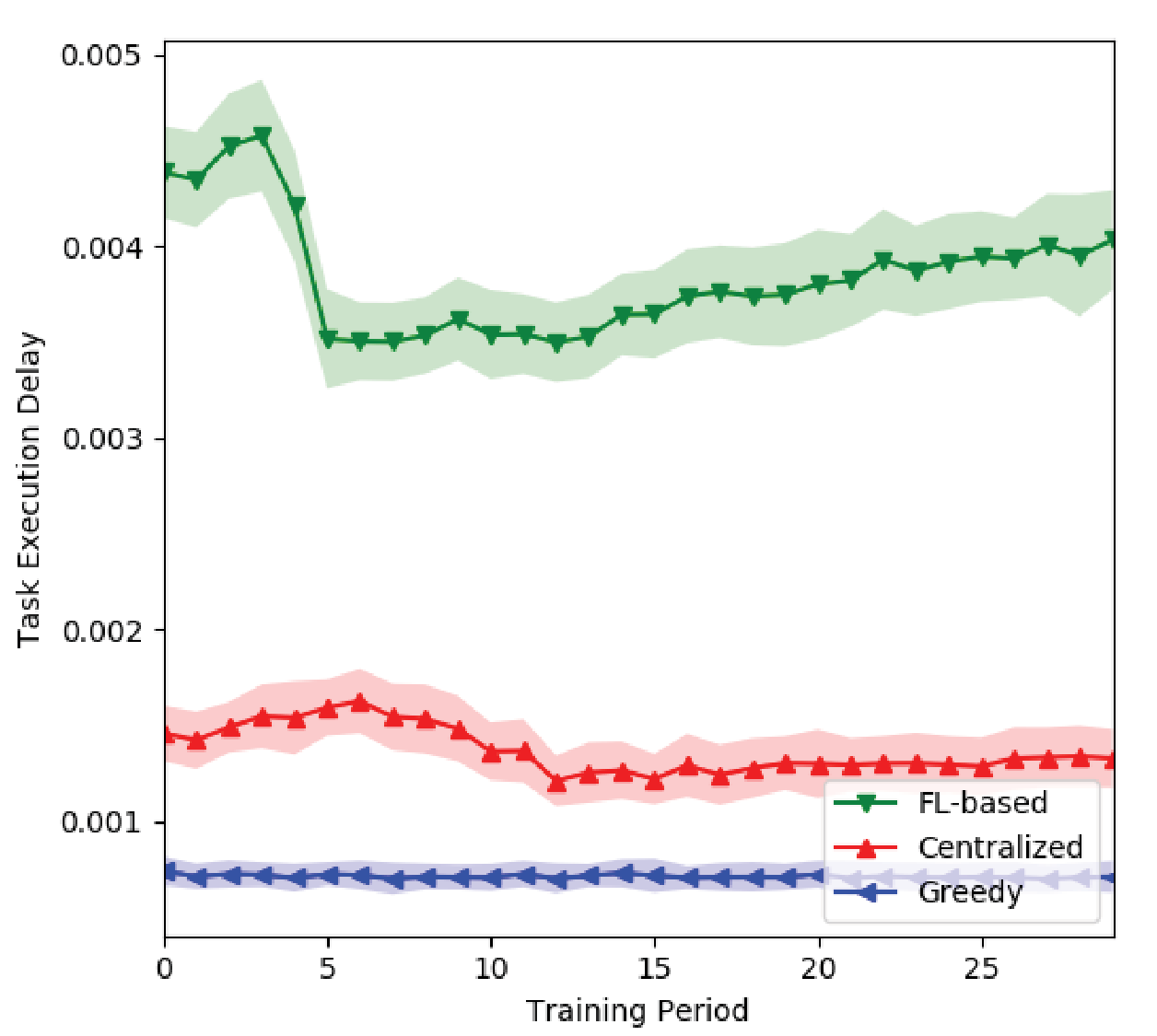}    
		\caption{Comparison of task execution delay.}    
		\label{fig:single4}
	\end{minipage}%
	\hfill    
	\begin{minipage}[t]{0.5\linewidth}    
		\centering    
		\includegraphics[width=1.0\textwidth]{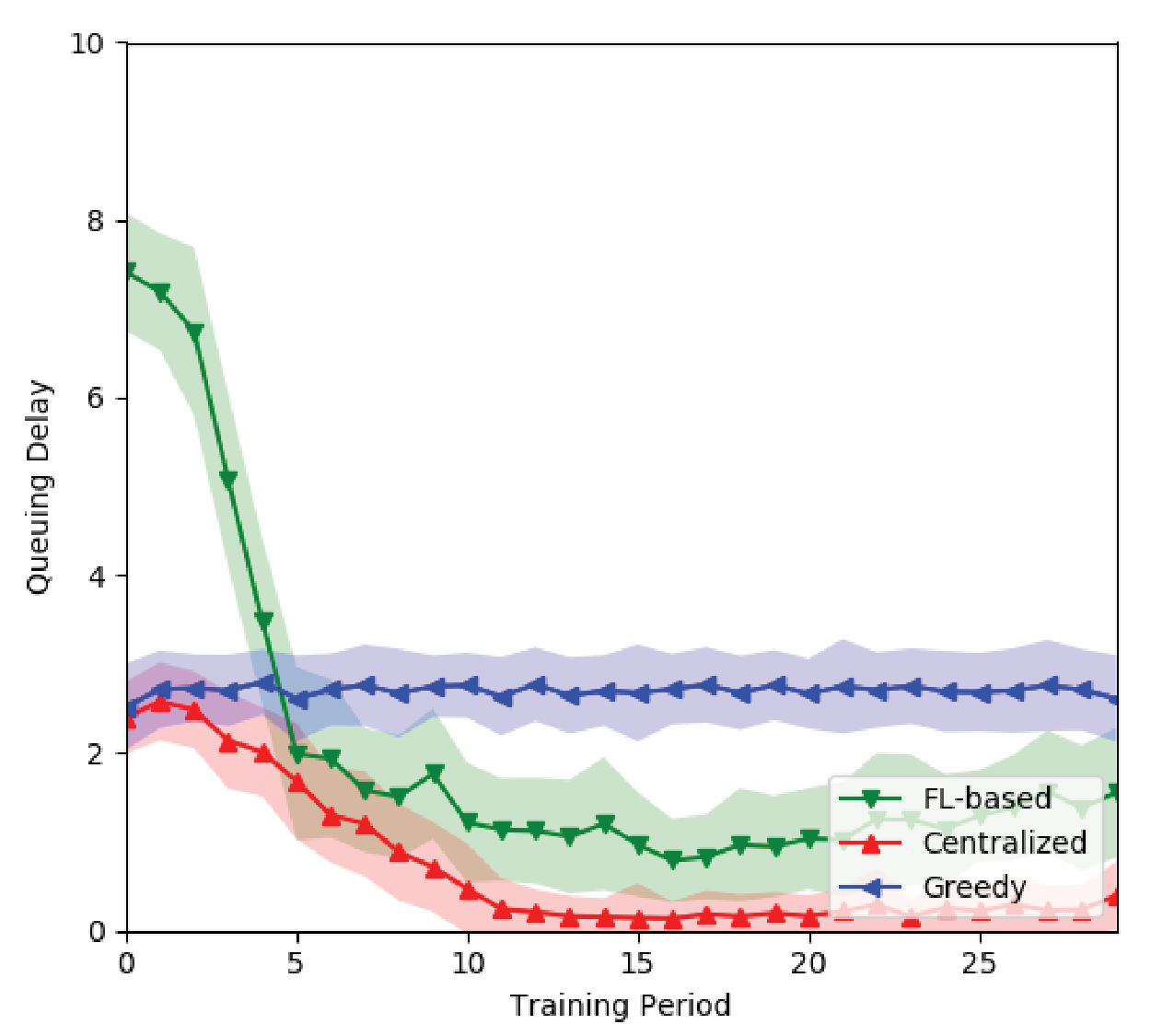}
		\caption{Comparison of queuing delay.}    
		\label{fig:single2}
	\end{minipage}    
\end{figure}

The queuing delay of the FL-based DRL training is generally lower than that of the centralized DRL training, but the queuing delay of the FL-based DRL training is higher overall. 

Next is the comparison of the number of task drop of FL-based DRL training and centralized DRL training. The number of task drop refers to the number of tasks lost when the number of newly generated tasks in the task queue reaches the upper limit. The number of task drop of FL-based DRL training and centralized DRL training is shown in Fig. \ref{fig:single3}.

The number of task drop of FL-based DRL training is much higher than the centralized DRL training in the early stage, but the number of task drop of FL-based DRL training decreases rapidly and ultimately stable not far from the centralized DRL trained. Overall, both of them have significantly decreased in the early stage which indicates the effectiveness of the training, but the number of task drop of FL-based DRL training will be higher than the centralized DRL training. This may be due to the queuing delay and the task execution delay of the FL-based DRL training are higher so that many tasks in the task queue cannot be completed in time. Therefore, the task queue in FL-based DRL training is more easily full, which leads to the loss of tasks.

\begin{figure}[htbp]    
	\begin{minipage}[t]{0.5\linewidth}    
		\centering    
		\includegraphics[width=1.0\textwidth]{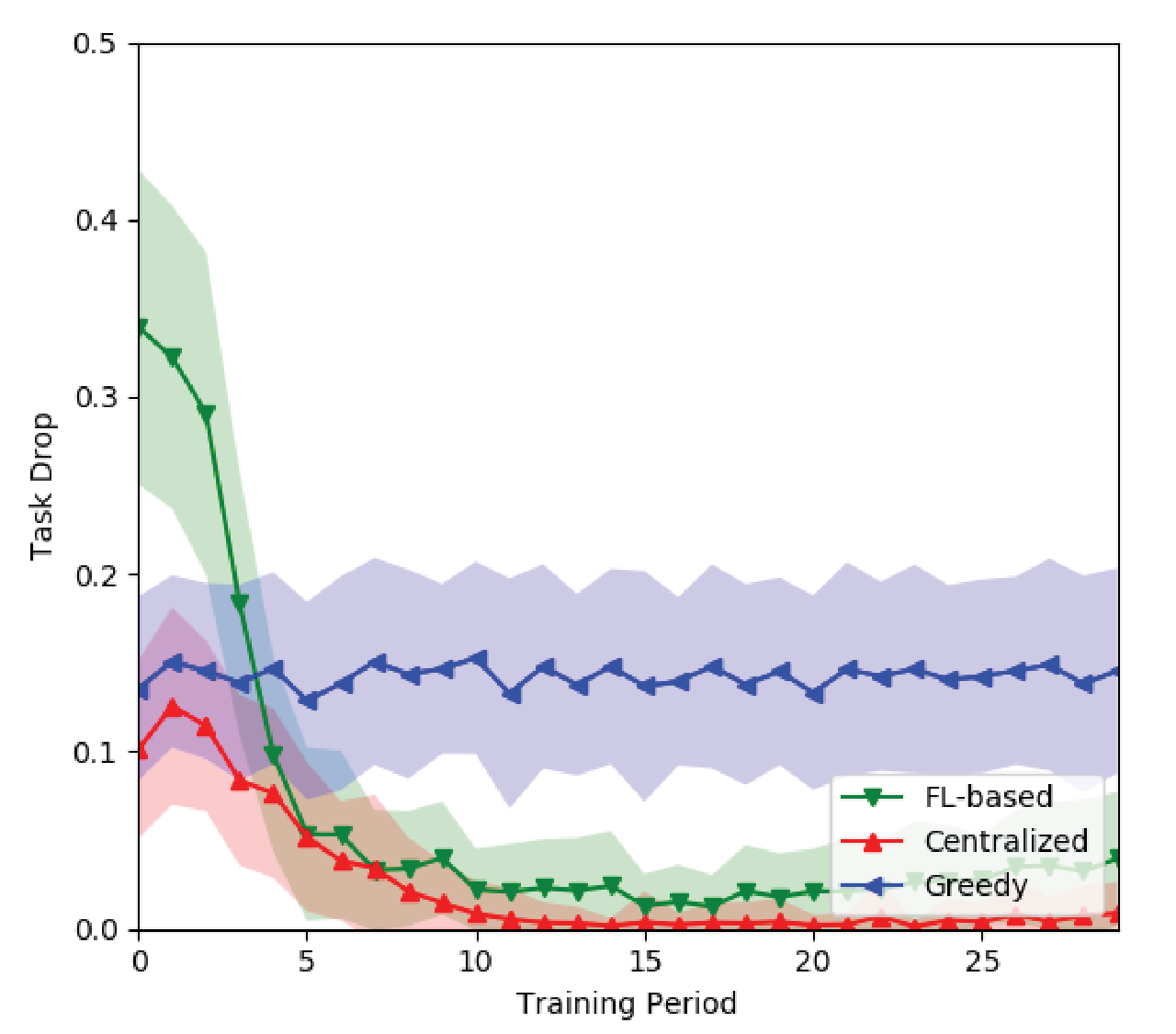}    
		\caption{Comparison of the number of task drop.}    
		\label{fig:single3}
	\end{minipage}%
	\hfill    
	\begin{minipage}[t]{0.5\linewidth}    
		\centering    
		\includegraphics[width=1.0\textwidth]{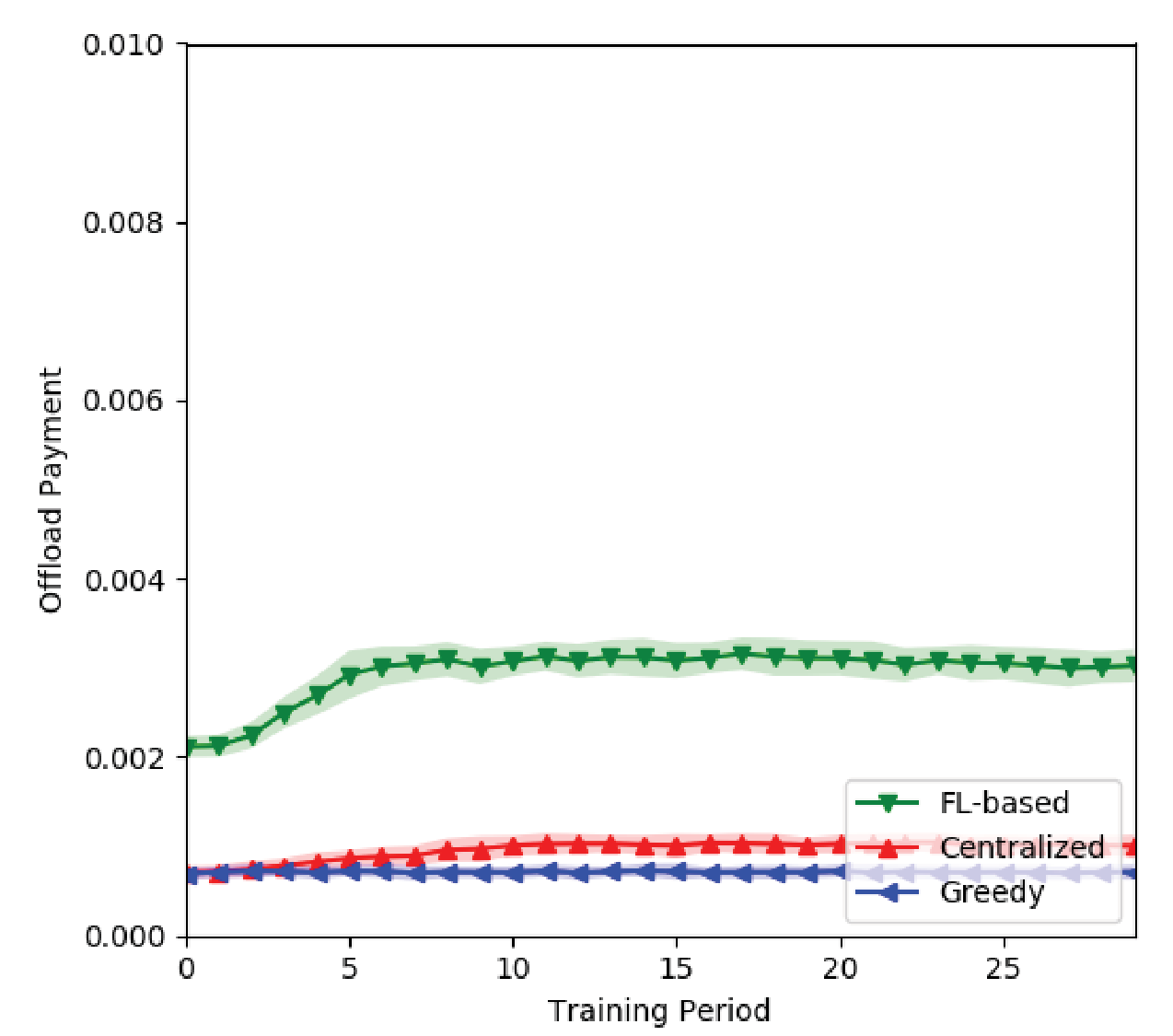}
		\caption{Comparison of offload payment.}    
		\label{fig:single1}
	\end{minipage}    
\end{figure}

Besides, due to the limited computing resources of EN, the offload payment of FL-based DRL training and centralized DRL training is also recorded. Offload payment refers to the payment required to use EN, and its value is positively related to the duration of the occupation. As shown in Fig. \ref{fig:single1}, the offload payment of FL-based DRL training and centralized DRL training are relatively stable overall, but FL-based DRL training takes up more time on EN.

This part of the experiment compares the centralized DRL training, FL-based DRL training and, the greedy strategy in several key parameters in detail. One of the prominent characteristics is that FL-based DRL training will transfer more tasks to the edge node for execution, resulting in the change of key parameters.

\subsection{Analysis of Computation Offloading Performance}
Since centralized DRL training is the benchmark for our proposed policy, we developed a further comparison with it.

Randomly soliciting three IoT devices for investigation, Fig. \ref{fig:client1_FL_utility}-\ref{fig:client3_FL_utility} and Fig. \ref{fig:centralized_training_utility} present the performance of FL-enabled DRL training and centralized DRL training, respectively. Accordingly, their training loss statics are correspondingly given in Fig. \ref{fig:client1_FL_loss}-\ref{fig:client3_FL_loss} and Fig. \ref{fig:centralized_training_loss}.
\begin{figure*}[htbp]%
	\centering
	\subfigure[Utility variation of device 1]
	{\includegraphics[width=4.2cm,height = 3.5cm]{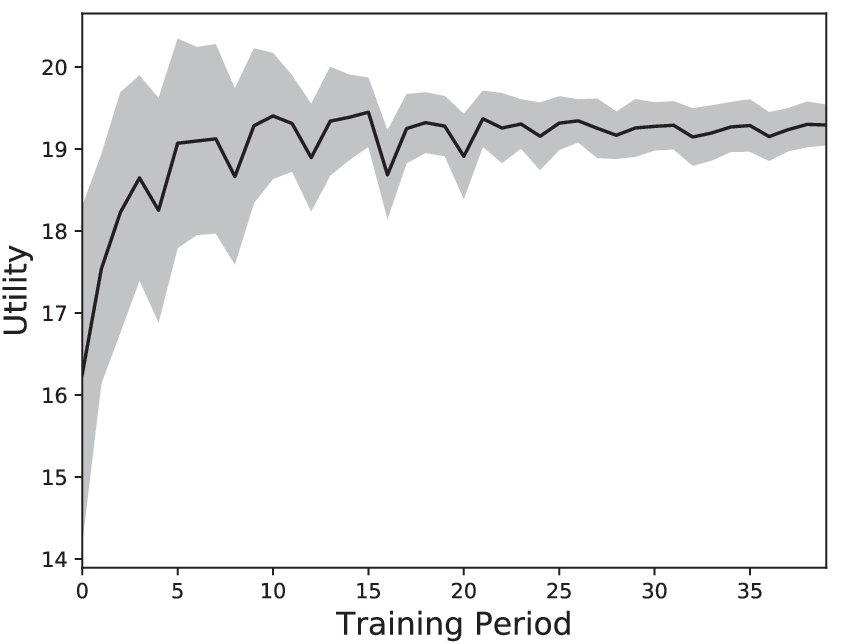}
		\label{fig:client1_FL_utility}
	}
	\subfigure[Utility variation of device 2]
	{\includegraphics[width=4.2cm,height = 3.5cm]{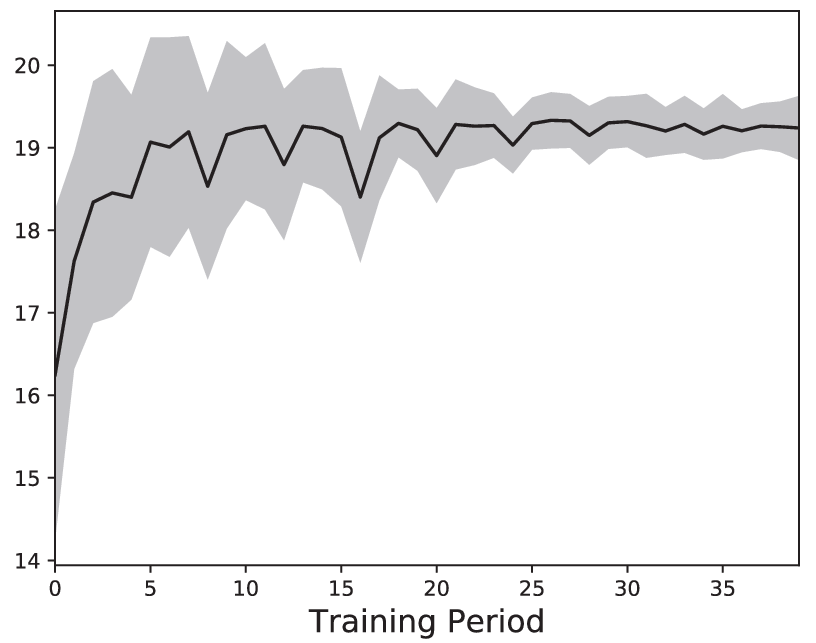}
		\label{fig:client2_FL_utility}
	}
	\subfigure[Utility variation of device 3]
	{\includegraphics[width=4.2cm,height = 3.5cm]{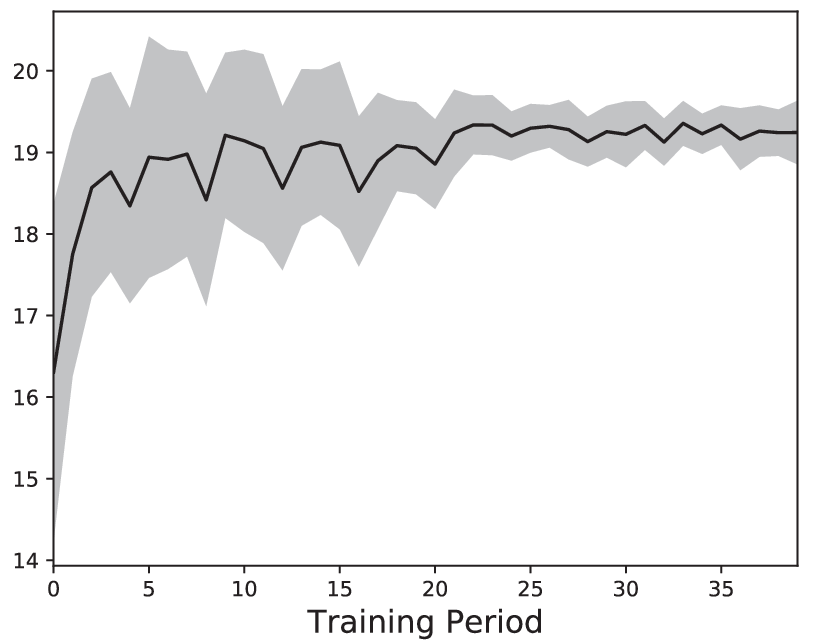}
		\label{fig:client3_FL_utility}
	}
	
	\subfigure[Loss variation of device 1]
	{\includegraphics[width=4.2cm,height = 3.5cm]{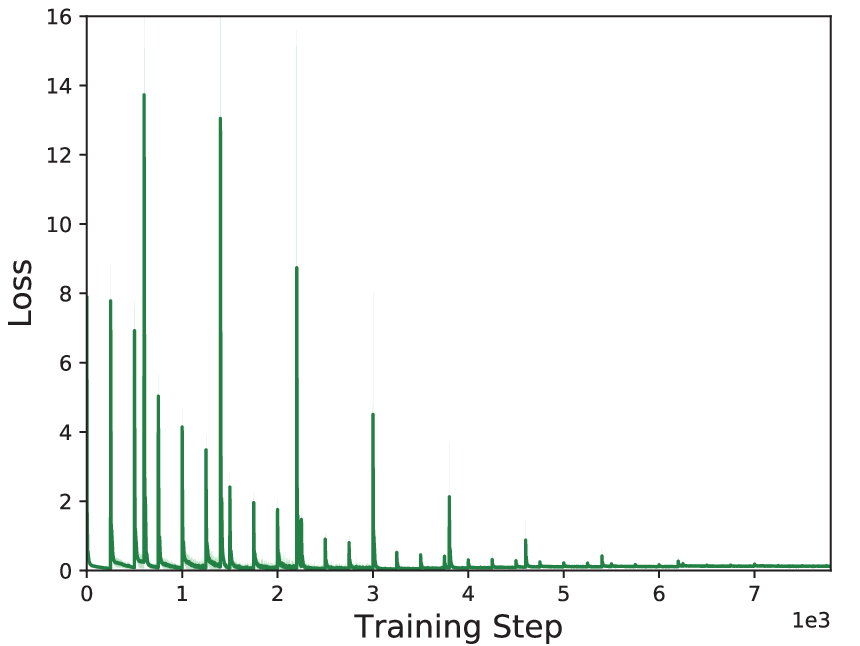}
		\label{fig:client1_FL_loss}
	}
	\subfigure[Loss variation of device 2]
	{\includegraphics[width=4.2cm,height = 3.5cm]{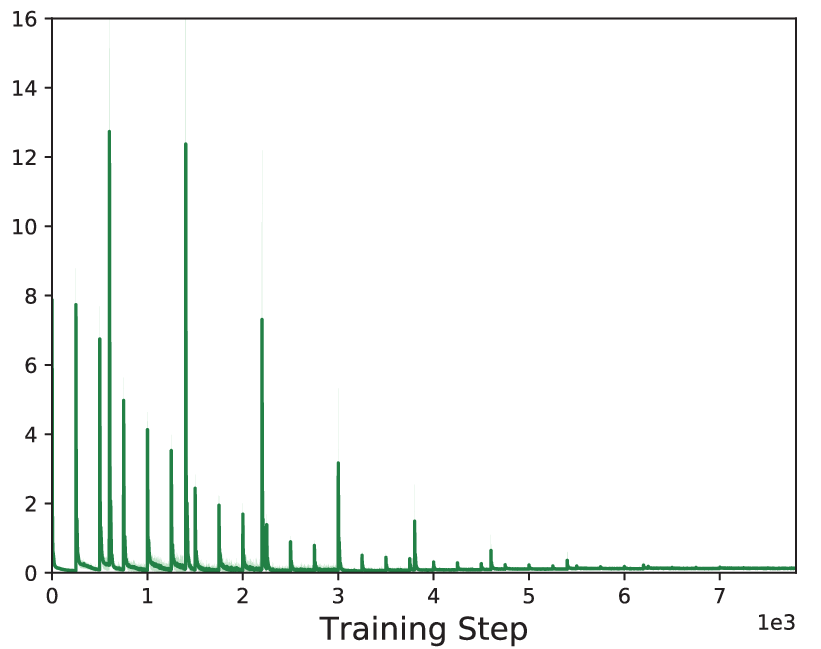}
		\label{fig:client2_FL_loss}
	}
	\subfigure[Loss variation of device 3]
	{\includegraphics[width=4.2cm,height = 3.5cm]{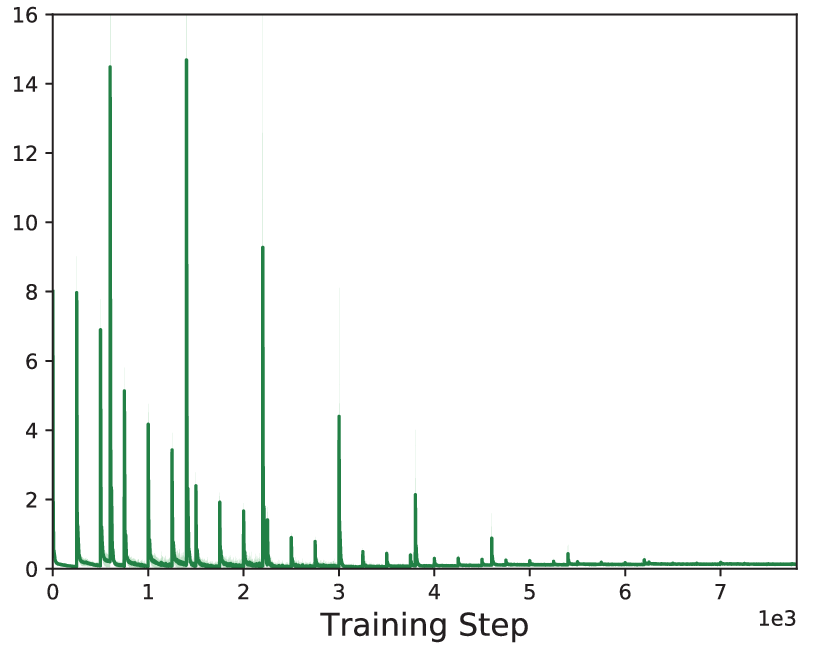}
		\label{fig:client3_FL_loss}
	}
	
	\caption{Computation offloading performance with Federated Learning-based. }
	\label{fig:comparison1}
\end{figure*}
\begin{figure*}[htbp]%
	\centering
	\subfigure[Utility variation of centralized training]
	{\includegraphics[width=0.47\linewidth]{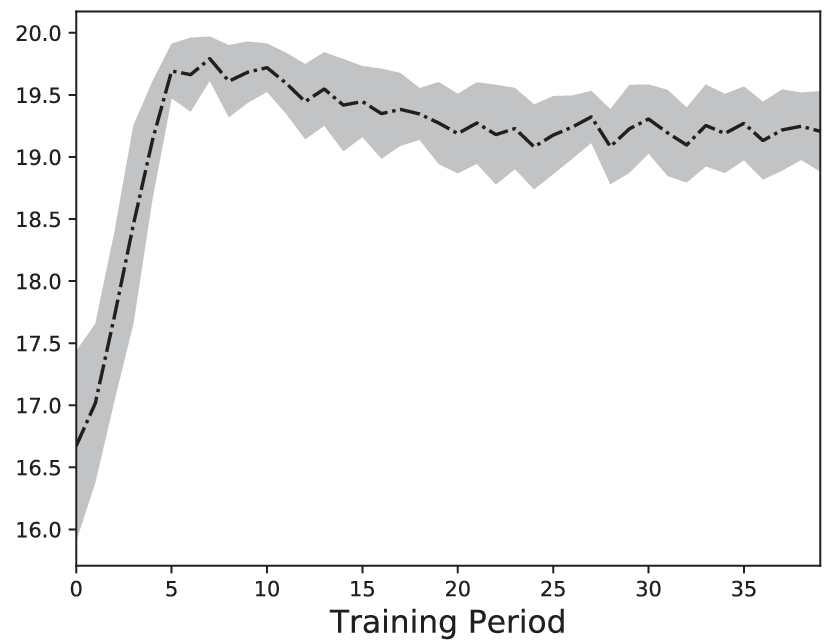}
		\label{fig:centralized_training_utility}
	}
	\subfigure[Loss variation of centralized training]
	{\includegraphics[width=0.47\linewidth]{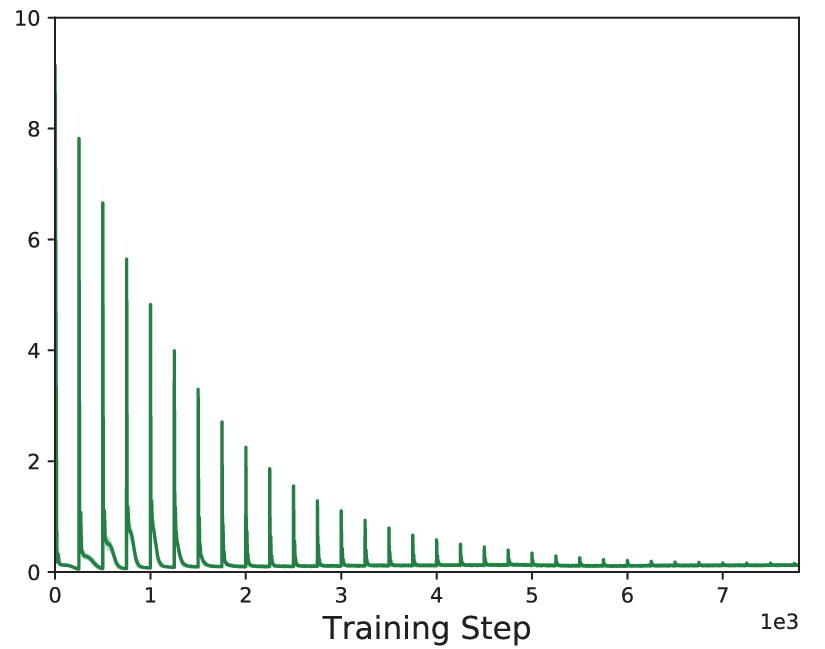}
		\label{fig:centralized_training_loss}
	}
	\caption{Computation offloading performance with centralized DRL training training.}
	\label{fig:comparison2}
\end{figure*}

The performance evaluation was as follows by a comparison of experimental results. \textbf{\textit{1)}} Obviously, the standard deviation of the training of the centralized training is less than the FL-based DRL training. This shows that the centralized training has better stability during training. In addition, as the training continues, the utility and standard deviations of FL-based DRL training continue to shrink, and eventually at the same level as centralized training. This experimental result verifies that FL-based DRL training achieves the same level of performance as centralized training and verifies its effectiveness. \textbf{\textit{2)}} Since it is assumed that the wireless channel in the centralized training can upload the training data to the EN without loss and does not cause the delay. In fact, this is impossible and this further proves the effectiveness of FL-based DRL training. Under this assumption, once trained for a period of time, FL-based DRL training can achieve the same level of performance as centralized training. Therefore, in the actual situation, the centralized training performance will be even better.

In addition, FL has two disadvantages. On the one hand, FL-based DRL training not only performs poorly during training but also requires a longer time to converge. On the other hand, FL-based DRL training is at the same level but relatively poor compared to centralized training. So how to fine-grain the scheduling of data transmission for optimization can be considered in future work.

\section{Conclusions}

In this paper, we investigate a computation offloading optimization problem. In more detail, each IoT device can make decisions about offloading tasks and allocating energy to maximize the expected long-term utility. For this, we propose an offloading algorithm based on FL and DRL and run in parallel across multiple IoT devices. On the one hand, DRL enables each IoT device to make decisions independently according to its own dynamic environment. On the other hand, FL further reduces the transmission consumption between IoT devices and EN and greatly enhances the privacy protection of data. In addition, we also carried out the necessary theoretical analysis and implemented a series of simulation experiments for this algorithm. The experimental results show that the proposed algorithm is applicable to various environments and verify the effectiveness of FL-based DRL. In the future, we will take a deep study in if there are model compression techniques with respect to DRL, and how to schedule the FL-based DRL training in a fine-grained fashion.

\bibliographystyle{ACM-Reference-Format}
\bibliography{tosn}


\begin{thebibliography}{39}


\ifx \showCODEN    \undefined \def \showCODEN     #1{\unskip}     \fi
\ifx \showDOI      \undefined \def \showDOI       #1{#1}\fi
\ifx \showISBNx    \undefined \def \showISBNx     #1{\unskip}     \fi
\ifx \showISBNxiii \undefined \def \showISBNxiii  #1{\unskip}     \fi
\ifx \showISSN     \undefined \def \showISSN      #1{\unskip}     \fi
\ifx \showLCCN     \undefined \def \showLCCN      #1{\unskip}     \fi
\ifx \shownote     \undefined \def \shownote      #1{#1}          \fi
\ifx \showarticletitle \undefined \def \showarticletitle #1{#1}   \fi
\ifx \showURL      \undefined \def \showURL       {\relax}        \fi
\providecommand\bibfield[2]{#2}
\providecommand\bibinfo[2]{#2}
\providecommand\natexlab[1]{#1}
\providecommand\showeprint[2][]{arXiv:#2}

\bibitem[\protect\citeauthoryear{Abbas, Zhang, Taherkordi, and Skeie}{Abbas
  et~al\mbox{.}}{2017}]%
        {abbas2017mobile}
\bibfield{author}{\bibinfo{person}{Nasir Abbas}, \bibinfo{person}{Yan Zhang},
  \bibinfo{person}{Amir Taherkordi}, {and} \bibinfo{person}{Tor Skeie}.}
  \bibinfo{year}{2017}\natexlab{}.
\newblock \showarticletitle{Mobile edge computing: A survey}.
\newblock \bibinfo{journal}{\emph{IEEE Internet of Things Journal}}
  \bibinfo{volume}{5}, \bibinfo{number}{1} (\bibinfo{year}{2017}),
  \bibinfo{pages}{450--465}.
\newblock


\bibitem[\protect\citeauthoryear{Adu-Manu, Adam, Tapparello, Ayatollahi, and
  Heinzelman}{Adu-Manu et~al\mbox{.}}{2018}]%
        {adu2018energy}
\bibfield{author}{\bibinfo{person}{Kofi~Sarpong Adu-Manu},
  \bibinfo{person}{Nadir Adam}, \bibinfo{person}{Cristiano Tapparello},
  \bibinfo{person}{Hoda Ayatollahi}, {and} \bibinfo{person}{Wendi Heinzelman}.}
  \bibinfo{year}{2018}\natexlab{}.
\newblock \showarticletitle{Energy-harvesting wireless sensor networks
  (EH-WSNs): A review}.
\newblock \bibinfo{journal}{\emph{ACM Transactions on Sensor Networks (TOSN)}}
  \bibinfo{volume}{14}, \bibinfo{number}{2} (\bibinfo{year}{2018}),
  \bibinfo{pages}{10}.
\newblock


\bibitem[\protect\citeauthoryear{Ananthanarayanan, Bahl, Bod{\'\i}k,
  Chintalapudi, Philipose, Ravindranath, and Sinha}{Ananthanarayanan
  et~al\mbox{.}}{2017}]%
        {ananthanarayanan2017real}
\bibfield{author}{\bibinfo{person}{Ganesh Ananthanarayanan},
  \bibinfo{person}{Paramvir Bahl}, \bibinfo{person}{Peter Bod{\'\i}k},
  \bibinfo{person}{Krishna Chintalapudi}, \bibinfo{person}{Matthai Philipose},
  \bibinfo{person}{Lenin Ravindranath}, {and} \bibinfo{person}{Sudipta Sinha}.}
  \bibinfo{year}{2017}\natexlab{}.
\newblock \showarticletitle{Real-time video analytics: The killer app for edge
  computing}.
\newblock \bibinfo{journal}{\emph{computer}} \bibinfo{volume}{50},
  \bibinfo{number}{10} (\bibinfo{year}{2017}), \bibinfo{pages}{58--67}.
\newblock


\bibitem[\protect\citeauthoryear{Bigmoyan}{Bigmoyan}{2017}]%
        {Keras}
\bibfield{author}{\bibinfo{person}{Bigmoyan}.} \bibinfo{year}{2017}\natexlab{}.
\newblock \bibinfo{title}{{Keras: The Python Deep Learning library}}.
\newblock
\newblock
\urldef\tempurl%
\url{https://keras.io}
\showURL{%
\tempurl}


\bibitem[\protect\citeauthoryear{Bonawitz, Eichner, Grieskamp, Huba, Ingerman,
  Ivanov, Kiddon, Konecny, Mazzocchi, McMahan, Overveldt, Petrou, Ramage, and
  Roselander}{Bonawitz et~al\mbox{.}}{2019}]%
        {bonawitz2019towards}
\bibfield{author}{\bibinfo{person}{Keith Bonawitz}, \bibinfo{person}{Hubert
  Eichner}, \bibinfo{person}{Wolfgang Grieskamp}, \bibinfo{person}{Dzmitry
  Huba}, \bibinfo{person}{Alex Ingerman}, \bibinfo{person}{Vladimir Ivanov},
  \bibinfo{person}{Chloe Kiddon}, \bibinfo{person}{Jakub Konecny},
  \bibinfo{person}{Stefano Mazzocchi}, \bibinfo{person}{H~Brendan McMahan},
  \bibinfo{person}{Timon~Van Overveldt}, \bibinfo{person}{David Petrou},
  \bibinfo{person}{Daniel Ramage}, {and} \bibinfo{person}{Jason Roselander}.}
  \bibinfo{year}{2019}\natexlab{}.
\newblock \showarticletitle{Towards federated learning at scale: System
  design}.
\newblock \bibinfo{journal}{\emph{arXiv preprint arXiv:1902.01046}}
  (\bibinfo{year}{2019}).
\newblock


\bibitem[\protect\citeauthoryear{Chen and Hao}{Chen and Hao}{2018}]%
        {Chen2018d}
\bibfield{author}{\bibinfo{person}{Min Chen} {and} \bibinfo{person}{Yixue
  Hao}.} \bibinfo{year}{2018}\natexlab{}.
\newblock \showarticletitle{Task Offloading for Mobile Edge Computing in
  Software Defined Ultra-Dense Network}.
\newblock \bibinfo{journal}{\emph{{IEEE} Journal on Selected Areas in
  Communications}} \bibinfo{volume}{36}, \bibinfo{number}{3}
  (\bibinfo{year}{2018}), \bibinfo{pages}{587--597}.
\newblock
\urldef\tempurl%
\url{https://doi.org/10.1109/JSAC.2018.2815360}
\showDOI{\tempurl}


\bibitem[\protect\citeauthoryear{Chen, Jiao, Li, and Fu}{Chen
  et~al\mbox{.}}{2016}]%
        {Chen2016f}
\bibfield{author}{\bibinfo{person}{Xu Chen}, \bibinfo{person}{Lei Jiao},
  \bibinfo{person}{Wenzhong Li}, {and} \bibinfo{person}{Xiaoming Fu}.}
  \bibinfo{year}{2016}\natexlab{}.
\newblock \showarticletitle{Efficient Multi-User Computation Offloading for
  Mobile-Edge Cloud Computing}.
\newblock \bibinfo{journal}{\emph{{IEEE/ACM} Trans. Netw.}}
  \bibinfo{volume}{24}, \bibinfo{number}{5} (\bibinfo{year}{2016}),
  \bibinfo{pages}{2795--2808}.
\newblock
\urldef\tempurl%
\url{https://doi.org/10.1109/TNET.2015.2487344}
\showDOI{\tempurl}


\bibitem[\protect\citeauthoryear{Chen, Zhang, Wu, Mao, Ji, and Bennis}{Chen
  et~al\mbox{.}}{2018}]%
        {Chen2018k}
\bibfield{author}{\bibinfo{person}{Xianfu Chen}, \bibinfo{person}{Honggang
  Zhang}, \bibinfo{person}{Celimuge Wu}, \bibinfo{person}{Shiwen Mao},
  \bibinfo{person}{Yusheng Ji}, {and} \bibinfo{person}{Mehdi Bennis}.}
  \bibinfo{year}{2018}\natexlab{}.
\newblock \showarticletitle{Optimized Computation Offloading Performance in
  Virtual Edge Computing Systems via Deep Reinforcement Learning}.
\newblock \bibinfo{journal}{\emph{IEEE Internet of Things Journal}}
  (\bibinfo{year}{2018}), \bibinfo{pages}{1--1}.
\newblock
\showISSN{2327-4662}
\urldef\tempurl%
\url{https://doi.org/10.1109/JIOT.2018.2876279}
\showDOI{\tempurl}


\bibitem[\protect\citeauthoryear{ETSI}{ETSI}{2014}]%
        {ETSI2}
\bibfield{author}{\bibinfo{person}{ETSI}.} \bibinfo{year}{2014}\natexlab{}.
\newblock \bibinfo{title}{{Mobile-Edge Computing–Introductory Technical White
  Paper}}.
\newblock
\newblock
\urldef\tempurl%
\url{https://portal.etsi.org/Portals/0/TBpages/MEC/Docs/Mobile-edge_Computing_-_Introductory_Technical_White_Paper_V1%2018-09-14.pdf}
\showURL{%
\tempurl}


\bibitem[\protect\citeauthoryear{Garey and Johnson}{Garey and Johnson}{1978}]%
        {garey1978strong}
\bibfield{author}{\bibinfo{person}{Michael~R Garey} {and}
  \bibinfo{person}{David~S Johnson}.} \bibinfo{year}{1978}\natexlab{}.
\newblock \showarticletitle{``Strong''NP-Completeness Results: Motivation,
  Examples, and Implications}.
\newblock \bibinfo{journal}{\emph{Journal of the ACM (JACM)}}
  \bibinfo{volume}{25}, \bibinfo{number}{3} (\bibinfo{year}{1978}),
  \bibinfo{pages}{499--508}.
\newblock


\bibitem[\protect\citeauthoryear{Gerards, Hurink, and Kuper}{Gerards
  et~al\mbox{.}}{2014}]%
        {gerards2014interplay}
\bibfield{author}{\bibinfo{person}{Marco~ET Gerards}, \bibinfo{person}{Johann~L
  Hurink}, {and} \bibinfo{person}{Jan Kuper}.} \bibinfo{year}{2014}\natexlab{}.
\newblock \showarticletitle{On the interplay between global DVFS and scheduling
  tasks with precedence constraints}.
\newblock \bibinfo{journal}{\emph{IEEE Trans. Comput.}} \bibinfo{volume}{64},
  \bibinfo{number}{6} (\bibinfo{year}{2014}), \bibinfo{pages}{1742--1754}.
\newblock


\bibitem[\protect\citeauthoryear{Google}{Google}{2017a}]%
        {tflite}
\bibfield{author}{\bibinfo{person}{Google}.} \bibinfo{year}{2017}\natexlab{a}.
\newblock \bibinfo{title}{{Deploy machine learning models on mobile and IoT
  devices}}.
\newblock
\newblock
\urldef\tempurl%
\url{https://www.tensorflow.org/lite}
\showURL{%
\tempurl}


\bibitem[\protect\citeauthoryear{Google}{Google}{2017b}]%
        {federatedlearning}
\bibfield{author}{\bibinfo{person}{Google}.} \bibinfo{year}{2017}\natexlab{b}.
\newblock \bibinfo{title}{{Federated Learning: Collaborative Machine Learning
  without Centralized Training Data}}.
\newblock
\newblock
\urldef\tempurl%
\url{https://research.googleblog.com/2017/04/federated-learning-collaborative.html}
\showURL{%
\tempurl}


\bibitem[\protect\citeauthoryear{Grewe, Wagner, Arumaithurai, Psaras, and
  Kutscher}{Grewe et~al\mbox{.}}{2017}]%
        {grewe2017information}
\bibfield{author}{\bibinfo{person}{Dennis Grewe}, \bibinfo{person}{Marco
  Wagner}, \bibinfo{person}{Mayutan Arumaithurai}, \bibinfo{person}{Ioannis
  Psaras}, {and} \bibinfo{person}{Dirk Kutscher}.}
  \bibinfo{year}{2017}\natexlab{}.
\newblock \showarticletitle{Information-centric mobile edge computing for
  connected vehicle environments: Challenges and research directions}. In
  \bibinfo{booktitle}{\emph{Proceedings of the Workshop on Mobile Edge
  Communications}}. ACM, \bibinfo{pages}{7--12}.
\newblock


\bibitem[\protect\citeauthoryear{Hu, Wang, and Ji}{Hu et~al\mbox{.}}{2013}]%
        {hu2013wireless}
\bibfield{author}{\bibinfo{person}{Xiaoya Hu}, \bibinfo{person}{Bingwen Wang},
  {and} \bibinfo{person}{Han Ji}.} \bibinfo{year}{2013}\natexlab{}.
\newblock \showarticletitle{A wireless sensor network-based structural health
  monitoring system for highway bridges}.
\newblock \bibinfo{journal}{\emph{Computer-Aided Civil and Infrastructure
  Engineering}} \bibinfo{volume}{28}, \bibinfo{number}{3}
  (\bibinfo{year}{2013}), \bibinfo{pages}{193--209}.
\newblock


\bibitem[\protect\citeauthoryear{Islam, Kwak, Kabir, Hossain, and Kwak}{Islam
  et~al\mbox{.}}{2015}]%
        {islam2015internet}
\bibfield{author}{\bibinfo{person}{SM~Riazul Islam}, \bibinfo{person}{Daehan
  Kwak}, \bibinfo{person}{MD~Humaun Kabir}, \bibinfo{person}{Mahmud Hossain},
  {and} \bibinfo{person}{Kyung-Sup Kwak}.} \bibinfo{year}{2015}\natexlab{}.
\newblock \showarticletitle{The internet of things for health care: a
  comprehensive survey}.
\newblock \bibinfo{journal}{\emph{IEEE Access}}  \bibinfo{volume}{3}
  (\bibinfo{year}{2015}), \bibinfo{pages}{678--708}.
\newblock


\bibitem[\protect\citeauthoryear{Jaakkola, Jordan, and Singh}{Jaakkola
  et~al\mbox{.}}{1994}]%
        {jaakkola1994convergence}
\bibfield{author}{\bibinfo{person}{Tommi Jaakkola}, \bibinfo{person}{Michael~I
  Jordan}, {and} \bibinfo{person}{Satinder~P Singh}.}
  \bibinfo{year}{1994}\natexlab{}.
\newblock \showarticletitle{Convergence of stochastic iterative dynamic
  programming algorithms}. In \bibinfo{booktitle}{\emph{Advances in neural
  information processing systems}}. \bibinfo{pages}{703--710}.
\newblock


\bibitem[\protect\citeauthoryear{Jia, Shelhamer, Donahue, Karayev, Long,
  Girshick, Guadarrama, and Darrell}{Jia et~al\mbox{.}}{2014}]%
        {jia2014caffe}
\bibfield{author}{\bibinfo{person}{Yangqing Jia}, \bibinfo{person}{Evan
  Shelhamer}, \bibinfo{person}{Jeff Donahue}, \bibinfo{person}{Sergey Karayev},
  \bibinfo{person}{Jonathan Long}, \bibinfo{person}{Ross Girshick},
  \bibinfo{person}{Sergio Guadarrama}, {and} \bibinfo{person}{Trevor Darrell}.}
  \bibinfo{year}{2014}\natexlab{}.
\newblock \showarticletitle{Caffe: Convolutional architecture for fast feature
  embedding}. In \bibinfo{booktitle}{\emph{Proceedings of the 22nd ACM
  international conference on Multimedia}}. ACM, \bibinfo{pages}{675--678}.
\newblock


\bibitem[\protect\citeauthoryear{Jie, Pei, Jun, Yun, and Wei}{Jie
  et~al\mbox{.}}{2013}]%
        {jie2013smart}
\bibfield{author}{\bibinfo{person}{Yin Jie}, \bibinfo{person}{Ji~Yong Pei},
  \bibinfo{person}{Li Jun}, \bibinfo{person}{Guo Yun}, {and}
  \bibinfo{person}{Xu Wei}.} \bibinfo{year}{2013}\natexlab{}.
\newblock \showarticletitle{Smart home system based on iot technologies}. In
  \bibinfo{booktitle}{\emph{2013 International Conference on Computational and
  Information Sciences}}. IEEE, \bibinfo{pages}{1789--1791}.
\newblock


\bibitem[\protect\citeauthoryear{Johnson and Garey}{Johnson and Garey}{1979}]%
        {johnson1979computers}
\bibfield{author}{\bibinfo{person}{David~S Johnson} {and}
  \bibinfo{person}{Michael~R Garey}.} \bibinfo{year}{1979}\natexlab{}.
\newblock \bibinfo{booktitle}{\emph{Computers and intractability: A guide to
  the theory of NP-completeness}}. Vol.~\bibinfo{volume}{1}.
\newblock \bibinfo{publisher}{WH Freeman San Francisco}.
\newblock


\bibitem[\protect\citeauthoryear{Kamruzzaman, Sarkar, Gutierrez, and
  Ray}{Kamruzzaman et~al\mbox{.}}{2017}]%
        {kamruzzaman2017study}
\bibfield{author}{\bibinfo{person}{Md Kamruzzaman}, \bibinfo{person}{Nurul~I
  Sarkar}, \bibinfo{person}{Jairo Gutierrez}, {and}
  \bibinfo{person}{Sayan~Kumar Ray}.} \bibinfo{year}{2017}\natexlab{}.
\newblock \showarticletitle{A study of IoT-based post-disaster management}. In
  \bibinfo{booktitle}{\emph{2017 International Conference on Information
  Networking (ICOIN)}}. IEEE, \bibinfo{pages}{406--410}.
\newblock


\bibitem[\protect\citeauthoryear{Khalil, Gnawali, Benhaddou, and
  Subhlok}{Khalil et~al\mbox{.}}{2018}]%
        {khalil2018sonicdoor}
\bibfield{author}{\bibinfo{person}{Nacer Khalil}, \bibinfo{person}{Omprakash
  Gnawali}, \bibinfo{person}{Driss Benhaddou}, {and} \bibinfo{person}{Jaspal
  Subhlok}.} \bibinfo{year}{2018}\natexlab{}.
\newblock \showarticletitle{SonicDoor: A Person Identification System Based on
  Modeling of Shape, Behavior, and Walking Patterns}.
\newblock \bibinfo{journal}{\emph{ACM Transactions on Sensor Networks (TOSN)}}
  \bibinfo{volume}{14}, \bibinfo{number}{3-4} (\bibinfo{year}{2018}),
  \bibinfo{pages}{27}.
\newblock


\bibitem[\protect\citeauthoryear{Kone{\v{c}}n{\`y}, McMahan, Ramage, and
  Richt{\'a}rik}{Kone{\v{c}}n{\`y} et~al\mbox{.}}{2016a}]%
        {konevcny2016federated2}
\bibfield{author}{\bibinfo{person}{Jakub Kone{\v{c}}n{\`y}},
  \bibinfo{person}{H~Brendan McMahan}, \bibinfo{person}{Daniel Ramage}, {and}
  \bibinfo{person}{Peter Richt{\'a}rik}.} \bibinfo{year}{2016}\natexlab{a}.
\newblock \showarticletitle{Federated optimization: Distributed machine
  learning for on-device intelligence}.
\newblock \bibinfo{journal}{\emph{arXiv preprint arXiv:1610.02527}}
  (\bibinfo{year}{2016}).
\newblock


\bibitem[\protect\citeauthoryear{Kone{\v{c}}n{\`y}, McMahan, Yu, Richt{\'a}rik,
  Suresh, and Bacon}{Kone{\v{c}}n{\`y} et~al\mbox{.}}{2016b}]%
        {konevcny2016federated}
\bibfield{author}{\bibinfo{person}{Jakub Kone{\v{c}}n{\`y}},
  \bibinfo{person}{H~Brendan McMahan}, \bibinfo{person}{Felix~X Yu},
  \bibinfo{person}{Peter Richt{\'a}rik}, \bibinfo{person}{Ananda~Theertha
  Suresh}, {and} \bibinfo{person}{Dave Bacon}.}
  \bibinfo{year}{2016}\natexlab{b}.
\newblock \showarticletitle{Federated learning: Strategies for improving
  communication efficiency}.
\newblock \bibinfo{journal}{\emph{arXiv preprint arXiv:1610.05492}}
  (\bibinfo{year}{2016}).
\newblock


\bibitem[\protect\citeauthoryear{Mao, You, Zhang, Huang, and Letaief}{Mao
  et~al\mbox{.}}{2017}]%
        {Mao2017A}
\bibfield{author}{\bibinfo{person}{Yuyi Mao}, \bibinfo{person}{Changsheng You},
  \bibinfo{person}{Jun Zhang}, \bibinfo{person}{Kaibin Huang}, {and}
  \bibinfo{person}{Khaled~B. Letaief}.} \bibinfo{year}{2017}\natexlab{}.
\newblock \showarticletitle{A Survey on Mobile Edge Computing: The
  Communication Perspective}.
\newblock \bibinfo{journal}{\emph{IEEE Communications Surveys \& Tutorials}}
  \bibinfo{volume}{19}, \bibinfo{number}{4} (\bibinfo{year}{2017}),
  \bibinfo{pages}{2322--2358}.
\newblock


\bibitem[\protect\citeauthoryear{McMahan, Moore, Ramage, Hampson, and
  y~Arcas}{McMahan et~al\mbox{.}}{2017}]%
        {mcmahan2016communication}
\bibfield{author}{\bibinfo{person}{Brendan McMahan}, \bibinfo{person}{Eider
  Moore}, \bibinfo{person}{Daniel Ramage}, \bibinfo{person}{Seth Hampson},
  {and} \bibinfo{person}{Blaise~Ag{\"{u}}era y Arcas}.}
  \bibinfo{year}{2017}\natexlab{}.
\newblock \showarticletitle{Communication-Efficient Learning of Deep Networks
  from Decentralized Data}. In \bibinfo{booktitle}{\emph{Proceedings of the
  20th International Conference on Artificial Intelligence and Statistics,
  {AISTATS} 2017, Fort Lauderdale, FL, {USA}}}
  \emph{(\bibinfo{series}{Proceedings of Machine Learning Research})},
  Vol.~\bibinfo{volume}{54}. \bibinfo{publisher}{{PMLR}},
  \bibinfo{pages}{1273--1282}.
\newblock
\urldef\tempurl%
\url{http://proceedings.mlr.press/v54/mcmahan17a.html}
\showURL{%
\tempurl}


\bibitem[\protect\citeauthoryear{Mnih, Kavukcuoglu, Silver, Rusu, Veness,
  Bellemare, Graves, Riedmiller, Fidjeland, Ostrovski, Petersen, Beattie,
  Sadik, Antonoglou, King, Kumaran, Wierstra, Legg, and Hassabis}{Mnih
  et~al\mbox{.}}{2015}]%
        {Mnih2015}
\bibfield{author}{\bibinfo{person}{Volodymyr Mnih}, \bibinfo{person}{Koray
  Kavukcuoglu}, \bibinfo{person}{David Silver}, \bibinfo{person}{Andrei~A
  Rusu}, \bibinfo{person}{Joel Veness}, \bibinfo{person}{Marc~G Bellemare},
  \bibinfo{person}{Alex Graves}, \bibinfo{person}{Martin Riedmiller},
  \bibinfo{person}{Andreas~K Fidjeland}, \bibinfo{person}{Georg Ostrovski},
  \bibinfo{person}{Stig Petersen}, \bibinfo{person}{Charles Beattie},
  \bibinfo{person}{Amir Sadik}, \bibinfo{person}{Ioannis Antonoglou},
  \bibinfo{person}{Helen King}, \bibinfo{person}{Dharshan Kumaran},
  \bibinfo{person}{Daan Wierstra}, \bibinfo{person}{Shane Legg}, {and}
  \bibinfo{person}{Demis Hassabis}.} \bibinfo{year}{2015}\natexlab{}.
\newblock \showarticletitle{Human-level control through deep reinforcement
  learning}.
\newblock \bibinfo{journal}{\emph{Nature}} \bibinfo{volume}{518},
  \bibinfo{number}{7540} (\bibinfo{year}{2015}), \bibinfo{pages}{529--533}.
\newblock
\urldef\tempurl%
\url{https://doi.org/10.1038/nature14236}
\showDOI{\tempurl}


\bibitem[\protect\citeauthoryear{PARLIAMENT and COUNCIL}{PARLIAMENT and
  COUNCIL}{2016}]%
        {GDPR}
\bibfield{author}{\bibinfo{person}{REGULATION (EU) 2016/679 OF THE~EUROPEAN
  PARLIAMENT} {and} \bibinfo{person}{OF~THE COUNCIL}.}
  \bibinfo{year}{2016}\natexlab{}.
\newblock \bibinfo{title}{{Mobile-Edge Computing–Introductory Technical White
  Paper}}.
\newblock
\newblock
\urldef\tempurl%
\url{https://eur-lex.europa.eu/legal-content/EN/TXT/?qid=1564833514064&uri=CELEX:32016R0679}
\showURL{%
\tempurl}


\bibitem[\protect\citeauthoryear{Pyattaev, Johnsson, Andreev, and
  Koucheryavy}{Pyattaev et~al\mbox{.}}{2015}]%
        {pyattaev2015communication}
\bibfield{author}{\bibinfo{person}{Alexander Pyattaev},
  \bibinfo{person}{Kerstin Johnsson}, \bibinfo{person}{Sergey Andreev}, {and}
  \bibinfo{person}{Yevgeni Koucheryavy}.} \bibinfo{year}{2015}\natexlab{}.
\newblock \showarticletitle{Communication challenges in high-density
  deployments of wearable wireless devices}.
\newblock \bibinfo{journal}{\emph{IEEE Wireless Communications}}
  \bibinfo{volume}{22}, \bibinfo{number}{1} (\bibinfo{year}{2015}),
  \bibinfo{pages}{12--18}.
\newblock


\bibitem[\protect\citeauthoryear{Rao, Lu, and Zhou}{Rao et~al\mbox{.}}{2017}]%
        {rao2017attention}
\bibfield{author}{\bibinfo{person}{Yongming Rao}, \bibinfo{person}{Jiwen Lu},
  {and} \bibinfo{person}{Jie Zhou}.} \bibinfo{year}{2017}\natexlab{}.
\newblock \showarticletitle{Attention-aware deep reinforcement learning for
  video face recognition}. In \bibinfo{booktitle}{\emph{Proceedings of the IEEE
  International Conference on Computer Vision}}. \bibinfo{pages}{3931--3940}.
\newblock


\bibitem[\protect\citeauthoryear{Samuel}{Samuel}{2016}]%
        {samuel2016review}
\bibfield{author}{\bibinfo{person}{S~Sujin~Issac Samuel}.}
  \bibinfo{year}{2016}\natexlab{}.
\newblock \showarticletitle{A review of connectivity challenges in IoT-smart
  home}. In \bibinfo{booktitle}{\emph{2016 3rd MEC International conference on
  big data and smart city (ICBDSC)}}. IEEE, \bibinfo{pages}{1--4}.
\newblock


\bibitem[\protect\citeauthoryear{Shah, Fiser, Faust, Kew, and Hakkani-Tur}{Shah
  et~al\mbox{.}}{2018}]%
        {shah2018follownet}
\bibfield{author}{\bibinfo{person}{Pararth Shah}, \bibinfo{person}{Marek
  Fiser}, \bibinfo{person}{Aleksandra Faust}, \bibinfo{person}{J~Chase Kew},
  {and} \bibinfo{person}{Dilek Hakkani-Tur}.} \bibinfo{year}{2018}\natexlab{}.
\newblock \showarticletitle{Follownet: Robot navigation by following natural
  language directions with deep reinforcement learning}.
\newblock \bibinfo{journal}{\emph{arXiv preprint arXiv:1805.06150}}
  (\bibinfo{year}{2018}).
\newblock


\bibitem[\protect\citeauthoryear{Shi, Cao, Zhang, Li, and Xu}{Shi
  et~al\mbox{.}}{2016}]%
        {Shi2016}
\bibfield{author}{\bibinfo{person}{Weisong Shi}, \bibinfo{person}{Jie Cao},
  \bibinfo{person}{Quan Zhang}, \bibinfo{person}{Youhuizi Li}, {and}
  \bibinfo{person}{Lanyu Xu}.} \bibinfo{year}{2016}\natexlab{}.
\newblock \showarticletitle{Edge Computing: Vision and Challenges}.
\newblock \bibinfo{journal}{\emph{{IEEE} Internet of Things Journal}}
  \bibinfo{volume}{3}, \bibinfo{number}{5} (\bibinfo{year}{2016}),
  \bibinfo{pages}{637--646}.
\newblock
\urldef\tempurl%
\url{https://doi.org/10.1109/JIOT.2016.2579198}
\showDOI{\tempurl}


\bibitem[\protect\citeauthoryear{Silver, Huang, Maddison, Guez, Sifre, Van
  Den~Driessche, Schrittwieser, Antonoglou, Panneershelvam, Lanctot, Dieleman,
  Grewe, Nham, Kalchbrenner, Sutskever, Lillicrap, Leach, Kavukcuoglu, Graepel,
  and Hassabis}{Silver et~al\mbox{.}}{2016}]%
        {silver2016mastering}
\bibfield{author}{\bibinfo{person}{David Silver}, \bibinfo{person}{Aja Huang},
  \bibinfo{person}{Chris~J Maddison}, \bibinfo{person}{Arthur Guez},
  \bibinfo{person}{Laurent Sifre}, \bibinfo{person}{George Van Den~Driessche},
  \bibinfo{person}{Julian Schrittwieser}, \bibinfo{person}{Ioannis Antonoglou},
  \bibinfo{person}{Veda Panneershelvam}, \bibinfo{person}{Marc Lanctot},
  \bibinfo{person}{Sander Dieleman}, \bibinfo{person}{Dominik Grewe},
  \bibinfo{person}{John Nham}, \bibinfo{person}{Nal Kalchbrenner},
  \bibinfo{person}{Ilya Sutskever}, \bibinfo{person}{Timothy Lillicrap},
  \bibinfo{person}{Madeleine Leach}, \bibinfo{person}{Koray Kavukcuoglu},
  \bibinfo{person}{Thore Graepel}, {and} \bibinfo{person}{Demis Hassabis}.}
  \bibinfo{year}{2016}\natexlab{}.
\newblock \showarticletitle{Mastering the game of Go with deep neural networks
  and tree search}.
\newblock \bibinfo{journal}{\emph{nature}} \bibinfo{volume}{529},
  \bibinfo{number}{7587} (\bibinfo{year}{2016}), \bibinfo{pages}{484}.
\newblock


\bibitem[\protect\citeauthoryear{Sutton and Barto}{Sutton and Barto}{2018}]%
        {sutton2018reinforcement}
\bibfield{author}{\bibinfo{person}{Richard~S Sutton} {and}
  \bibinfo{person}{Andrew~G Barto}.} \bibinfo{year}{2018}\natexlab{}.
\newblock \bibinfo{booktitle}{\emph{Reinforcement learning: An introduction}}.
\newblock \bibinfo{publisher}{MIT press}.
\newblock


\bibitem[\protect\citeauthoryear{van Hasselt, Guez, and Silver}{van Hasselt
  et~al\mbox{.}}{2016}]%
        {Hasselt20016}
\bibfield{author}{\bibinfo{person}{Hado van Hasselt}, \bibinfo{person}{Arthur
  Guez}, {and} \bibinfo{person}{David Silver}.}
  \bibinfo{year}{2016}\natexlab{}.
\newblock \showarticletitle{Deep Reinforcement Learning with Double
  Q-Learning}. In \bibinfo{booktitle}{\emph{Proceedings of the Thirtieth {AAAI}
  Conference on Artificial Intelligence,2016, Phoenix, Arizona, {USA.}}}
  \bibinfo{publisher}{{AAAI} Press}, \bibinfo{pages}{2094--2100}.
\newblock
\urldef\tempurl%
\url{http://www.aaai.org/ocs/index.php/AAAI/AAAI16/paper/view/12389}
\showURL{%
\tempurl}


\bibitem[\protect\citeauthoryear{Wang, Zhang, Liu, Bhuiyan, and Jin}{Wang
  et~al\mbox{.}}{2018}]%
        {wang2018secure}
\bibfield{author}{\bibinfo{person}{Tian Wang}, \bibinfo{person}{Guangxue
  Zhang}, \bibinfo{person}{Anfeng Liu}, \bibinfo{person}{Md~Zakirul~Alam
  Bhuiyan}, {and} \bibinfo{person}{Qun Jin}.} \bibinfo{year}{2018}\natexlab{}.
\newblock \showarticletitle{A secure IoT service architecture with an efficient
  balance dynamics based on cloud and edge computing}.
\newblock \bibinfo{journal}{\emph{IEEE Internet of Things Journal}}
  (\bibinfo{year}{2018}).
\newblock


\bibitem[\protect\citeauthoryear{Yang, Wu, Yin, Li, and Zhao}{Yang
  et~al\mbox{.}}{2017}]%
        {yang2017survey}
\bibfield{author}{\bibinfo{person}{Yuchen Yang}, \bibinfo{person}{Longfei Wu},
  \bibinfo{person}{Guisheng Yin}, \bibinfo{person}{Lijie Li}, {and}
  \bibinfo{person}{Hongbin Zhao}.} \bibinfo{year}{2017}\natexlab{}.
\newblock \showarticletitle{A survey on security and privacy issues in
  Internet-of-Things}.
\newblock \bibinfo{journal}{\emph{IEEE Internet of Things Journal}}
  \bibinfo{volume}{4}, \bibinfo{number}{5} (\bibinfo{year}{2017}),
  \bibinfo{pages}{1250--1258}.
\newblock


\bibitem[\protect\citeauthoryear{Zeydan, Bastug, Bennis, Kader, Karatepe, Er,
  and Debbah}{Zeydan et~al\mbox{.}}{2016}]%
        {Zeydan2016}
\bibfield{author}{\bibinfo{person}{Engin Zeydan}, \bibinfo{person}{Ejder
  Bastug}, \bibinfo{person}{Mehdi Bennis}, \bibinfo{person}{Manhal~Abdel
  Kader}, \bibinfo{person}{Ilyas~Alper Karatepe}, \bibinfo{person}{Ahmet~Salih
  Er}, {and} \bibinfo{person}{Merouane Debbah}.}
  \bibinfo{year}{2016}\natexlab{}.
\newblock \showarticletitle{{Big data caching for networking: moving from cloud
  to edge}}.
\newblock \bibinfo{journal}{\emph{IEEE Communications Magazine}}
  \bibinfo{volume}{54}, \bibinfo{number}{9} (\bibinfo{date}{sep}
  \bibinfo{year}{2016}), \bibinfo{pages}{36--42}.
\newblock
\showISSN{0163-6804}
\urldef\tempurl%
\url{https://doi.org/10.1109/MCOM.2016.7565185}
\showDOI{\tempurl}


\end{thebibliography}

\end{document}